\pdfoutput=1
\documentclass{article}
\usepackage{iclr2027_conference,times}
\iclrfinalcopy
\usepackage[T1]{fontenc}
\usepackage[utf8]{inputenc}
\usepackage{amsmath,amssymb,amsthm,mathtools}
\usepackage{graphicx,booktabs,array,tabularx,longtable}
\usepackage{xcolor}
\usepackage[ruled,vlined,linesnumbered]{algorithm2e}
\usepackage{enumitem}
\usepackage{microtype}
\usepackage{needspace}
\usepackage[section]{placeins}
\usepackage{hyperref}
\usepackage[capitalize,nameinlink]{cleveref}
\hypersetup{colorlinks=true,linkcolor=blue,citecolor=blue,urlcolor=blue,
 pdftitle={Edge-Private Matching Kernels through Local Decoding},
 pdfauthor={Michael Dinitz, George Z. Li, Quanquan C. Liu, and Felix Zhou}}
\newtheorem{theorem}{Theorem}
\newtheorem{lemma}[theorem]{Lemma}

\newtheorem{proposition}[theorem]{Proposition}
\theoremstyle{definition}
\newtheorem{definition}[theorem]{Definition}

\newcommand{\eps}{\varepsilon}
\newcommand{\E}{\mathbb E}
\newcommand{\Prb}{\mathbb P}
\newcommand{\R}{\mathbb R}

\newcommand{\Lap}{\operatorname{Lap}}
\newcommand{\AT}{\mathsf{AboveThreshold}}

\newcommand{\Seq}{\mathsf{SequentialMatching}}
\newcommand{\Par}{\mathsf{ParallelMatching}}
\newcommand{\ind}{\mathbf 1}

\DeclarePairedDelimiter{\ceil}{\lceil}{\rceil}
\DeclarePairedDelimiter{\floor}{\lfloor}{\rfloor}
\crefname{algocf}{Algorithm}{Algorithms}
\Crefname{algocf}{Algorithm}{Algorithms}
\SetKwInput{KwIn}{Input}
\SetKwInput{KwOut}{Output}
\DontPrintSemicolon
\SetAlCapFnt{\normalsize}
\SetAlFnt{\normalsize}
\SetAlgoNlRelativeSize{0}
\setlist{leftmargin=*,nosep}
\allowdisplaybreaks
\title{Edge-Private Matching Kernels\\through Local Decoding}
\author{%
\begin{tabular}[t]{@{}p{0.5\textwidth}@{}l@{}}
Michael Dinitz & George Z. Li \\
\normalfont Johns Hopkins University & \normalfont Carnegie Mellon University \\
\normalfont\texttt{mdinitz@cs.jhu.edu} & \normalfont\texttt{gzli@andrew.cmu.edu} \\[2.5ex]
Quanquan C. Liu & Felix Zhou \\
\normalfont Yale University & \normalfont Yale University \\
\normalfont\texttt{quanquan.liu@yale.edu} & \normalfont\texttt{felix.zhou@yale.edu}
\end{tabular}%
}
\date{}
\begin{document}
\maketitle
\lhead{Preprint}
\begin{abstract}
A matching is a set of edges with no shared endpoints. We study how participants can retain useful matching structure when each knows its own relationships but those relationships are sensitive to outside observers. Our protocols publish a pure edge-differentially-private record of coordination messages. Each participant combines that record with its private neighbor list to obtain a shortlist of relationships; the two endpoints agree on every selected relationship. For a graph with $n$ vertices, the main protocol uses $O(n)$ communication rounds, and its completed public record can be stored in $O(n\log n)$ bits. Let $0<\eps\le1$ be the privacy budget, $0<\rho<1/4$ the allowed failure probability, and $2+\eta$, with $0<\eta\le1$, the approximation factor. For a universal constant $c_0$, a per-vertex capacity of at least $c_0\log(n/\rho)/(\eta\eps)$ suffices for two guarantees to hold together with probability at least $1-\rho$: the selected graph retains a $1/(2+\eta)$ fraction of the best edge count under the same capacity, and it contains an ordinary matching of at least that fraction of the original maximum. The key property is that every omitted edge touches a vertex whose shortlist is nearly full. We also prove lower bounds explaining why local decoding and a probabilistic capacity guarantee are needed, characterize a class of consistent local decoders, and give parallel and insertion-only streaming extensions. The released public messages satisfy differential privacy; the selected relationships remain a secret at their endpoints.
\end{abstract}
\section{Introduction}\label{sec:intro}
A graph can describe which participants are eligible to work together. A \emph{matching} selects disjoint pairs where no participant appears in more than one selected edge. A \emph{maximum matching} contains as many pairs as possible. Its size changes by at most one when an edge is added or removed, so differential privacy permits accurate estimation of that number \citep{dworkroth2014}. Retaining the pairs themselves is a different task: they identify the relationships we wish to protect.

We study a setting in which each participant already knows its own relationships. The goal is to coordinate shortlists of eligible partners. Edge differential privacy limits how much adding or removing one relationship changes the distribution of the public messages. This leads to our central question: \emph{Can differentially  private communication coordinate small, mutually consistent shortlists that preserve many possible disjoint assignments?}

\paragraph{Public candidates, private selections.}
Let $G=(V,E)$ be the relationship graph, with $n=|V|$ participants. An edge in $E$ is an unordered pair of participants who know that they are connected. The \emph{public transcript} $B$ is the complete record of coordination messages and public random choices. Following the billboard model \citep{hsu2016}, it specifies a \emph{candidate graph} $C(B)$ of pairs that may be selected, including pairs absent from $E$.

For a vertex $v$, let $N_G(v)$ be its private set of neighbors and $C(B)(v)$ its public set of candidate neighbors. Its shortlist $M(v)$ keeps only neighbors appearing in both sets. The union of the selected edges is denoted by $M$:
\begin{equation}\label{eq:intro-decoder}
 M(v)=N_G(v)\cap C(B)(v),\qquad M=E(G)\cap C(B).
\end{equation}
For example, a vertex with true neighbors $\{u,w\}$ and candidate neighbors $\{w,z\}$ keeps only $w$. A pair is a candidate for both endpoints or neither, so the endpoints agree without exchanging their shortlists. Only $B$ is published. The candidate graph may be dense; the capacity bound applies to the true selected edges in $M$. We call this procedure \emph{local intersection decoding}.

\paragraph{What should the shortlists preserve?}
A capacity $b$ allows each participant at most $b$ selected neighbors. An edge set satisfying this constraint is a \emph{$b$-matching}; an ordinary matching is the case $b=1$. We seek two properties. First, the number of selected edges should approximate the best possible edge count under capacity $b$. Second, the selected graph should contain a large ordinary matching. These properties differ: a star has many edges, all sharing its center, but supports only one disjoint pair. A bounded-degree subgraph that preserves a large matching is the \emph{matching kernel} studied here.

\paragraph{Main result and its mechanism.}
Our sequential protocol obtains both approximations within a factor $2+\eta$, where $0<\eta\le1$ controls closeness to a factor of two. Let $\eps$ be the privacy budget (smaller is stronger) and $\rho$ the allowed probability of a capacity or utility failure. A capacity proportional to $\log(n/\rho)/(\eta\eps)$ suffices. The protocol uses $O(n)$ communication rounds and has an $O(n\log n)$-bit stored transcript. \Cref{thm:main} states the precise guarantees.

The construction maintains one useful property: every omitted edge has an endpoint with a nearly full shortlist. We call this \emph{near-capacity saturation}. A vertex privately tests successively smaller public candidate sets, called \emph{masks}, and publishes only a stopping index. Removing one candidate at a time controls how much the true selected degree can fall between tests. A charging argument and the matching-kernel lemma of \citet{arar2018} turn this property into the two approximation guarantees.

\paragraph{Why this output model?}
Our lower bounds separate publishing a sparse candidate graph from privately decoding a sparse graph. They also distinguish a capacity bound that holds with probability one from one that may fail with probability $\rho$. A representation theorem shows that intersection decoding captures all deterministic, always-valid, mutually consistent decoders that use only a vertex's neighbor list and a pure-private transcript, on all simple graphs over the same public vertex set. Parallel and streaming variants explore different communication and access requirements.

A concrete interpretation is a reciprocal standby directory among computing sites with private bilateral agreements. Each site records selected partners locally, completing a planning task without uploading its directory. Matching preservation certifies that many disjoint pairings remain possible. Choosing or activating a particular pairing is a later task; \Cref{app:workflow} applies our approach to a concrete workflow.

\subsection{Related work}\label{sec:related}
Implicit optimization \citep{gupta2010} and private allocations \citep{hsu2016} coordinate public messages and private local outputs. The allocation settings use private valuations and supply constraints in bipartite markets; our undirected graphs impose capacities at both endpoints. Matching kernels originate with \citet{bhattacharya2015}; we use the preservation lemma of \citet[Lemma~2.4]{arar2018}. Our decoders access local neighbor lists; general implicit representations may use the full input \citep{kamath2026}.

Local edge privacy supports graph release and analysis \citep{qin2017,imola2021,dhulipala2022,eden2025}. We use single-hit AboveThreshold \citep{dworkroth2014,lyu2017} and composition for interleaved interactions \citep{vadhanwang2021,vadhanzhang2023}. Other graph objectives estimate matching size \citep{blocki2022}, preserve aggregates \citep{elias2020,liu2024}, or protect weights on public topologies \citep{sealfon2016,imola2026}. Our streaming extensions use continual observation \citep{dwork2010,chan2011,fichtenberger2021,jain2024,raskhodnikova2025} to protect the complete release sequence.

\paragraph{Preliminary version.}
A preliminary version of some of these results appeared in an earlier preprint by the same authors \citep{dinitz2025differentially}.
\section{Model, Notation, and Main Guarantee}\label{sec:model}
We consider simple undirected graphs: there are no self-loops or repeated edges. The vertex set $V$ is public and fixed, with $n\ge2$ vertices; the edge set $E(G)$ and its size $m=|E(G)|$ are private. The set $\binom V2$ denotes all unordered pairs of distinct vertices, including pairs that are not edges. A \emph{row} or \emph{adjacency list} at $v$ means the private neighbor set $N_G(v)$.

For any edge set $F$, $|F|$ is its number of edges, $d_F(v)$ is the number incident to $v$, and $\Delta(F)=\max_v\left(d_F(v)\right)$ is its maximum degree. We identify an edge set with the graph it forms on $V$; thus $|M|$ always counts selected edges, not vertices. For an integer $k\ge1$, define
\[
 \nu_k(G)=\max\{|F|:F\subseteq E(G),\ \Delta(F)\le k\}.
\]
This is the largest possible number of edges in a $k$-matching. In particular, $\nu_1(G)$ is the ordinary maximum matching size. The symbols $b$ and $k$ have different roles: $b$ bounds the algorithm's output degree, while $k$ is a comparison capacity or a saturation target. \emph{Same-capacity} approximation compares the output with $\nu_b(G)$.

Unless stated otherwise, $0<\eps\le1$, $0<\rho<1/4$, and $0<\eta\le1$. Here $\eps$ controls privacy, $\rho$ controls the probability of failure of the structural guarantees, and $\eta$ controls their approximation factor $2+\eta$. All are public, as are vertex orders and capacities. Logarithms are natural unless their base is shown; $\ceil{x}$ and $\floor{x}$ round up and down to integers. We write $\Prb$ for probability and $\E$ for expectation over the mechanism's randomness. \Cref{app:notation} collects the notation by topic.

\begin{definition}[Edge privacy and local protocols]\label[definition]{def:privacy}
Graphs are \emph{edge-neighboring} if adding or removing one edge changes one into the other. A randomized mechanism $A$ is $(\eps,\delta)$-edge differentially private (DP) if every neighboring pair $G,G'$ and every measurable set $S$ of possible outputs satisfy
\[
 \Prb[A(G)\in S]\le e^\eps\Prb[A(G')\in S]+\delta.
\]
Pure privacy means $\delta=0$. A protocol is $\eps$-local-edge-private ($\eps$-LEDP) if each vertex uses only its adjacency list and public messages, and its entire public transcript, including public randomness, is pure $\eps$-edge DP \citep{qin2017,imola2021,dhulipala2022,eden2025}.
\end{definition}
The inequality limits how much observing an output can distinguish the presence of one edge. The approximate-privacy allowance $\delta$ is different from the utility-failure probability $\rho$; our algorithms use $\delta=0$. An edge changes both endpoints' lists, and the privacy analysis accounts for both.

\begin{definition}[Intersection-based decoding]\label[definition]{def:implicit}
A transcript $B$ specifies a public candidate graph $C(B)\subseteq\binom V2$ without another query to the private input. Write $C$ when the transcript is understood. The selected graph is $M(G,B)=E(G)\cap C(B)$, decoded separately at each vertex as in \eqref{eq:intro-decoder}. A degree-$b$ guarantee constrains $\Delta(M)$, not $\Delta(C(B))$. Feasibility (the degree bound) and utility (the approximation bounds) hold jointly with the stated probability. Privacy concerns the full transcript distribution, without conditioning on success.
\end{definition}

\paragraph{The main theorem.}
The first bound below limits every shortlist to $b$ neighbors. The second retains a fraction of the best capacity-$b$ edge count. The third preserves a fraction of the largest ordinary matching inside the selected graph. It certifies that those disjoint pairs exist; the protocol's output is the shortlists. The noise allowance $h_u$ and the integer degree allowance $D_4$ below set the precise sufficient capacity. A geometric-mask variant, described in \Cref{sec:geometric}, has a logarithmic selection threshold for fixed accuracy.
\begin{theorem}[Same-capacity approximation and matching preservation]\label{thm:main}
Let $0<\eta\le1$, $h_u=(16/\eps)\log(16n(n+2)/\rho)$, and $D_4=\ceil{4h_u}$. For every integer
\begin{equation}\label{eq:main-cap}
 b\ge \ceil{\frac{(2+\eta)D_4+2}{\eta}},
\end{equation}
a pure $\eps$-LEDP protocol uses $O(n)$ rounds and, with probability at least $1-\rho$, simultaneously guarantees
\begin{equation}\label{eq:main-guarantees}
 \Delta(M)\le b,\qquad |M|\ge\frac{\nu_b(G)}{2+\eta},\qquad
 \nu_1(M)\ge\frac{\nu_1(G)}{2+\eta}.
\end{equation}
Its complete communication transcript has an $O(n\log n)$-bit encoding for fixed public parameters and vertex order. The simpler sufficient bound $b\ge2^{10}\log(n/\rho)/(\eta\eps)$ implies \eqref{eq:main-cap}. Whenever $b\ge n-1$, the data-independent choice $C=\binom V2$ already gives all three guarantees with zero privacy loss. A geometric-mask variant gives \eqref{eq:main-guarantees} under
\begin{equation}\label{eq:geometric-main-cap}
 b\ge2^{13}\frac{\log(n/(\eta\rho))}{\eta\eps},
\end{equation}
with $O(\log n/\eta)$ rather than $n$ local selection tests per proposer. Both variants use $O(n)$ synchronized rounds.
\end{theorem}

\paragraph{Who sees what?}
Observers see $B$ and can reconstruct $C(B)$. A participant also knows its own row and shortlist. For an edge change at $\{u,v\}$, the transcript together with the shortlists of all other vertices obeys the same DP inequality: their private rows are unchanged (\Cref{lem:billboard}). Publishing the endpoints' selected edges would reveal additional private information and requires a separate analysis.

A communication \emph{round} is a synchronized exchange of public messages, allowing local computation between exchanges. A stored transcript is an encoding of the completed exchange. These are separate resources: the compact encoding does not determine how many bits are sent while the protocol runs (\Cref{tab:resources}).

\section{Private Selection and the Sequential Construction}\label{sec:sequential}
\subsection{Selecting a public set using a private count}
The algorithm builds the candidate graph one proposer at a time. A proposer should add enough true neighbors to make progress, but not exceed its capacity. It can count its true neighbors privately; publicly listing those neighbors would defeat the purpose. The Public Vertex Subset Mechanism (PVSM) instead chooses from public sets and reports which set it chose.

\paragraph{The private threshold test.}
A \emph{sensitivity-one query} changes by at most one when one input edge changes. Single-hit AboveThreshold, written $\AT_a$, compares such queries to a noisy threshold using privacy budget $a>0$. It adds one random variable from $\Lap(2/a)$ to the threshold and a fresh random variable from $\Lap(4/a)$ to each query. Here $\Lap(t)$ is the mean-zero Laplace distribution with scale $t$. A positive response means the noisy query reaches the noisy threshold; the mechanism stops at the first positive response. The complete interaction is $a$-DP regardless of the number of preceding negatives \citep{dworkroth2014,lyu2017}. All our query sequences have a finite length (\Cref{lem:AT}).

\paragraph{Applying the test to candidate sets.}
During the algorithm, write $C$ for the current candidate graph and $C(v)$ for its neighbors at proposer $v$. Let $S$ be a public set of eligible new candidates, disjoint from $C(v)$. The private count $x=|N_G(v)\cap C(v)|$ is the proposer's selected degree before this invocation. A \emph{mask} $A_r$ is a public subset of $S$; its index $r$ orders a nested sequence $A_0=S\supseteq A_1\supseteq\cdots$. Write $Z_r=|N_G(v)\cap A_r|$ for the true neighbors in that mask. Selecting it would give total degree $x+Z_r$.

With a public noise allowance $h>0$, PVSM tests $q_r=-(x+Z_r)$ against threshold $-(b-2h)$ and releases the first positive index. The minus signs turn an above-threshold test into a test for a sufficiently small total degree. If every test is negative, PVSM selects the empty set. Since old candidates and new candidates are disjoint, $x+Z_r$ counts neighbors in one public union; its sensitivity is one. Neither the counts nor the noises are released.

If all threshold and query noises have magnitude at most $h$, a positive test guarantees $x+Z_r\le b$. Conversely, any reached mask with $x+Z_r\le b-4h$ must pass. The $2h$ uncertainty comes from adding one threshold noise and one query noise. The mask schedules below choose $h$ so these facts hold together with probability at least $1-\rho$.

\paragraph{Coordinating the vertices.}
An \emph{active} vertex can propose or receive candidates. Following the per-vertex threshold design of \citet{dhulipala2025}, each vertex keeps a \emph{persistent saturation test} across iterations, with public threshold $\tau_{\rm sat}$. Its first positive answer permanently deactivates that vertex. The next active proposer selects among later active vertices. This order prevents duplicate pairs and limits each receiver to one new true edge per iteration.
\begin{algorithm}[ht]
\caption{$\Seq$: sequential public-subset coordination}\label{alg:sequential}
\KwIn{Public order $v_1,\ldots,v_n$, cap $b$, margin $h$, saturation threshold $\tau_{\rm sat}$, mask schedule, $a=\eps/4$; private adjacency lists.}
\KwOut{Public candidate graph $C$, represented by masks, indices, and active flags.}
Initialize $C\gets\varnothing$ and mark every vertex active\;
At each vertex initialize a persistent $\AT_a$ at threshold $\tau_{\rm sat}$\;
\For{$i=1,\ldots,n$}{
Every active $u$ queries its persistent mechanism on $|N_G(u)\cap C(u)|$; publish its answer; a positive answer permanently deactivates $u$\;
\If{$v_i$ is active}{
Let $S$ be the active vertices later than $v_i$ in the public order\;
Form the prescribed public masks on $S$ and run PVSM at $v_i$\;
Publish the outcome; add pairs $\{v_i,u\}$ from its selected mask to $C$\;
}
}
Each vertex privately decodes $N_G(v)\cap C(v)$\;
\end{algorithm}

\paragraph{Privacy accounting.}
Each vertex uses one persistent saturation mechanism and at most one PVSM, both with budget $a=\eps/4$. For any fixed public history, all candidate sets are fixed and their private counts have sensitivity one. An edge change affects two vertices, each spending at most $2a$, for total privacy loss $2(2a)=\eps$. \Cref{lem:endpoints} proves this for the complete interleaved exchange, including all public negative responses. The selection tests run locally; only their stopping outcome needs communication.

\subsection{Unit-step masks: why omitted edges imply saturation}\label{sec:unit-step}
Order $S$ by public identifiers, and let $A_r$ remove its first $r$ vertices. For instance, $\{u,w,z\}$ becomes $\{w,z\}$, then $\{z\}$, then the empty set. Each removal changes the true-neighbor count by zero or one. Test indices $0,\ldots,n-1$, with $A_r=\varnothing$ once $r\ge|S|$. Publish the first positive index or the distinct symbol $\bot$ if all tests are negative. Even two empty masks have distinct indices; the public active set determines whether a proposer is skipped. All masks are determined by the public order, so they require no random pair labels.

Choose the following parameters; $L_u$ sets a joint tail bound and $h_u$ is its noise allowance:
\begin{equation}\label{eq:unit-params}
 L_u=\log\left(\frac{16n(n+2)}\rho\right),\quad a=\eps/4,\quad
 h_u=16L_u/\eps,\quad \tau_{\rm sat}=b-2h_u.
\end{equation}
Thus the saturation threshold is $b-2h_u$ and the PVSM threshold is its negative. A \emph{prefix} means the execution through some completed iteration. The saturation target $k$ in the next statement is a lower bound on selected degree, not an upper bound on the shortlist.
\begin{proposition}[Unit-step prefix saturation]\label[proposition]{prop:unit-step}
Use \eqref{eq:unit-params} in \Cref{alg:sequential}. For integers $1\le k\le n-1$ and $b\ge k+D_4$, where $D_4=\ceil{4h_u}$, the protocol is pure $\eps$-LEDP and uses $O(n)$ rounds. With probability at least $1-\rho$, every prefix has degree at most $b$ and, after iteration $j$, every omitted true edge with an endpoint among $v_1,\ldots,v_j$ has an endpoint of current selected degree at least $k$. Hence $|M|\ge\nu_k(G)/2$.
\end{proposition}

\paragraph{Why the invariant holds.}
On the event that every noise is bounded by $h_u$, a negative saturation response implies degree strictly below $b$. An active receiver can therefore accept its at-most-one new edge. A positive response implies degree at least $b-4h_u$, and a positive selection test keeps the proposer's new degree at most $b$.

Now consider a proposer whose degree is below $k$. Choosing index zero keeps all eligible true neighbors. If it chooses index $r\ge1$, the previous mask was rejected as too large, and removing one candidate loses at most one true neighbor. Hence
\begin{equation}\label{eq:unit-progress}
 x+Z_r>b-4h_u-1\ge b-D_4-1.
\end{equation}
The new degree is an integer, so it is at least $b-D_4\ge k$. An empty mask must pass when the proposer remains below $k$, ruling out an all-negative outcome in that case. Therefore a proposer either keeps every eligible neighbor or reaches degree $k$. If an edge were omitted while both endpoints remained below $k$, its earlier endpoint would have selected it. This proves the stated property at every prefix. \Cref{app:unit-full} gives the probability bound and all boundary cases.

\subsection{Turning saturation into the two approximation guarantees}
Let $H$ be any subgraph of $G$ with the same vertices. The next lemma separates the deterministic graph argument from privacy: once saturation holds, both utility bounds follow without further randomization.
\begin{lemma}[Two consequences of saturation]\label[lemma]{lem:structural}
Let $H\subseteq G$ and suppose every omitted edge has an endpoint of $H$-degree at least an integer $k\ge1$. Then $|E(H)|\ge\nu_k(G)/2$. If also $\Delta(H)\le b$ and $k\le b$, then
\begin{equation}\label{eq:kernel}
 \nu_1(H)\ge\frac{k}{2(b+1)}\nu_1(G).
\end{equation}
For integers $1\le k\le b$, the benchmark capacities also satisfy
\begin{equation}\label{eq:capacity-compare}
 \nu_k(G)\ge\frac{k}{b+1}\nu_b(G).
\end{equation}
\end{lemma}

The first claim compares the retained edges with an optimal $k$-matching by assigning each omitted comparison edge to a saturated endpoint. Equation~\eqref{eq:kernel} is the matching-kernel lemma of \citet[Lemma~2.4]{arar2018}. To obtain \eqref{eq:capacity-compare}, partition an optimal $b$-matching into at most $b+1$ ordinary matchings and keep the $k$ largest. Such an edge-coloring groups edges so that no two in a group share a vertex. Their union has degree at most $k$. \Cref{app:combinatorial} proves these statements in full.

For the main theorem, take $k=b-D_4$ when $b<n-1$. Both utility fractions are at least $(b-D_4)/(2(b+1))$, and
\begin{equation}\label{eq:unit-exact-cap}
 \frac{b-D_4}{2(b+1)}\ge\frac1{2+\eta}
 \quad\Longleftrightarrow\quad
 \eta b\ge(2+\eta)D_4+2.
\end{equation}
This is exactly the theorem's capacity condition. When $b\ge n-1$, selecting every potential pair as a candidate retains the whole input with zero privacy loss. \Cref{app:main-conversion-proof} verifies the explicit constant and integer rounding.

For illustration, $n=100000$, $\eps=1$, $\eta=0.1$, and $\rho=0.01$ give $h_u\simeq486.4580774$ and sufficient capacity $b=40886$. This instantiates the worst-case bound; it is not a measured minimum shortlist size. The full calculation is in \Cref{app:unit-comparison}.

\subsection{Storing the completed public record}\label{sec:compact}
Most public saturation responses are repeated negatives. A vertex deactivates at most once and proposes at most once, so two records per vertex suffice to recover the interaction.
\begin{proposition}[Lossless unit-step transcript encoding]\label[proposition]{prop:compact}
Fix the public parameters and vertex order. The complete logical transcript of the unit-step protocol, on every realization, admits an encoding of
\begin{equation}\label{eq:compact-bits}
 n\ceil{\log_2(n+1)}+n\ceil{\log_2(n+2)}\quad\text{bits}.
\end{equation}
It reconstructs all public saturation answers, proposer outcomes, active sets, masks, and candidate pairs. It is a postprocessing of the transcript.
\end{proposition}

Store each vertex's first deactivation iteration (or ``never'') and its proposer outcome (skipped, all-negative, or one of the $n$ indices). These choices have $n+1$ and $n+2$ possibilities, explaining the two terms in \eqref{eq:compact-bits}. Replay reconstructs the active sets and hence all masks and messages, even on an execution where the utility guarantee fails. This is \emph{postprocessing}: it uses only information already in the public record. The direct online protocol can still send $n^2$ saturation answers. \Cref{app:compact} supplies the lossless encoder and decoder.

\subsection{Geometric masks and the selection horizon}\label{sec:geometric}
The geometric variant randomly thins each candidate set rather than deleting one vertex at a time. A spacing parameter $s\in(0,1/4]$ sets the inclusion probability $(1+s)^{-r}$ at mask index $r$. Shared public randomness makes these masks nested. With $J=\ceil{\log_{1+s}n}$, a proposal tests $J+2$ masks, including a terminal empty set. For $s=\eta/16$, this is $O(\log n/\eta)$ local selection tests instead of $n$.

Several true neighbors can disappear in one step, so the proof bounds their loss jointly across all indices before selecting a mask. \Cref{lem:thinning} gives this bound. The geometric construction uses saturation threshold $b-6h$ and the capacity in \eqref{eq:geometric-main-cap}; its parameters and proof are in \Cref{app:geometric-ledger,app:geometric-reduction}. Both variants retain persistent saturation tests over $n$ iterations.

\section{Why the Output Representation Matters}\label{sec:lower}
The positive results keep true selections at the endpoints and allow a probability $\rho$ of exceeding the degree bound. We now examine what happens when either requirement changes.

\subsection{Publishing a sparse candidate graph}
An \emph{explicit} mechanism publishes a candidate graph and bounds that public graph's degree, rather than the degree after local decoding. Its utility is the number of true input edges among the candidates; additional non-edges can help privacy but consume the degree budget. A \emph{perfect matching} pairs all vertices. Even on this simple input family, retaining a constant fraction of true edges under pure privacy with bounded $\eps$ requires an average expected degree budget growing linearly with $n$.

\Cref{thm:explicit} in \Cref{app:explicit-proof} gives the precise bound, including approximate privacy, fractional edge weights, and expected budgets. The proof relabels vertices randomly so all non-edges of a perfect matching receive equal expected weight. Switching two matching edges to crossed pairs then turns a true edge into a non-edge in four edits. Privacy along those edits bounds its inclusion. This obstruction concerns actual retained relationships, not a numerical graph statistic.

\subsection{Requiring the degree bound with probability one}
This subsection imposes a stronger requirement than \Cref{thm:main}: an \emph{almost-sure} degree cap holds with probability one on every input. The proposition is a conditional lower bound, not a guarantee of our algorithms. Let $\gamma\ge1$ be a finite multiplicative approximation factor.
\begin{proposition}[Zero-failure feasibility forces sparse candidates]\label[proposition]{prop:zero-failure}
Suppose a pure $\eps$-edge-private intersection decoder satisfies $\Delta(M(G,B))\le b$ almost surely on every graph on $V$. Then $\Delta(C(B))\le b$ almost surely on every graph. If $n\ge4$ is even and $\Prb[|M|\ge\nu_1(G)/\gamma]\ge1-\rho$ on every input, then
\begin{equation}\label{eq:zero-failure}
 b\ge\frac{(1-\rho)(n-2)}{\gamma e^{4\eps}}.
\end{equation}
\end{proposition}

On the complete graph every candidate is a true edge, so $M=C(B)$. An almost-sure bound on $M$ therefore rules out dense candidate graphs on this input. Pure privacy transfers every probability-zero transcript event across a finite sequence of edge edits. Hence it rules out dense candidates on every input, and the explicit theorem applies. \Cref{app:zero-short-proof} gives the full proof.

\subsection{Allowing a small failure probability}
Local decoding allows a much smaller degree, but still has a cost depending on privacy and failure probability. The next theorem uses a \emph{star}: one center joined to $q$ leaves. The integer $q$ specifies the number of edges in the hard input.
\begin{theorem}[Degree cost of intersection decoding]\label{thm:implicit-lower}
Let $\gamma\ge1$ be finite and $0<\rho\le e^{-2\eps}/16$. Suppose a pure $\eps$-edge-private billboard satisfies
$\Prb[|M(G,B)|\ge\nu_1(G)/\gamma]\ge1-\rho$ on every graph. Set
\[
 q=\min\left\{n-1,\floor{\frac{\log(1/\rho)}{2\eps}}\right\}.
\]
There is a $q$-edge star on which the decoded center degree is at least $q/2$ with probability at least $1/2$.
\end{theorem}

A finite multiplicative guarantee on a single-edge input must keep that edge whenever it succeeds. Privacy makes its omission unlikely even among the candidates on an empty input. Combining these requirements across possible edges and transferring them to a star forces high decoded degree (\Cref{app:implicit-proof}). If the algorithm also promises degree at most $b$ with failure less than $1/2$, it must have $b\ge q/2$. For fixed accuracy and failure $\rho=n^{-c}$ with fixed $c>0$, this matches the sequential logarithmic degree order, up to the ceiling imposed by $n-1$ possible neighbors.

\subsection{When must a local decoder be an intersection?}
A decoder is a rule mapping the public transcript and one vertex's private neighbor list to a shortlist. Write this rule as $D_v(N_G(v),B)$. It is \emph{valid} if it selects only true neighbors and \emph{endpoint-consistent} if both endpoints agree on every selected pair. The following theorem characterizes fixed deterministic rules of this kind. ``Measurable'' ensures their outputs define valid probability events.
\begin{theorem}[Representation of deterministic consistent decoders]\label{thm:representation}
Let $B(G)$ be pure $\eps$-edge DP on all simple graphs over the fixed finite set $V$. Each vertex uses a fixed deterministic measurable decoder $D_v(N_G(v),B)$, with no other retained private state or private randomness. Suppose on every graph, almost surely, $D_v(N_G(v),B)\subseteq N_G(v)$ and
\[
 u\in D_v(N_G(v),B)\quad\Longleftrightarrow\quad
 v\in D_u(N_G(u),B).
\]
Then a measurable public graph $C(B)$ satisfies $D_v(N_G(v),B)=N_G(v)\cap C(B)(v)$ almost surely for every graph and vertex. The selected edge sets may have arbitrary degree.
\end{theorem}

Here probability one concerns validity and endpoint agreement, not degree or approximation. Pure privacy transfers null transcript events across the finite graph domain. Removing their union leaves one set on which the deterministic decoders agree on every graph. Comparing a star with a single-edge graph then gives the representation. \Cref{app:representation} proves the claim and gives an approximate-privacy counterexample. The two pure-privacy barriers therefore apply throughout this decoder class.

\section{Parallel and Streaming Constructions}\label{sec:extensions}
The sequential protocol trades linear synchronization for a strong approximation with the same capacity as the benchmark. Two extensions explore other priorities: fewer communication rounds and updates as the input grows. The former remains local; the streaming algorithms use a trusted central curator who sees the incoming records. In both cases, only the prescribed public messages are released.

\subsection{Parallel coordination}
\begin{theorem}[Parallel benchmark and constant matching preservation]\label{thm:parallel}
For every integer $1\le k\le n-1$, an $\eps$-LEDP protocol uses $O(\log(n/\rho))$ rounds and, with probability at least $1-\rho$, simultaneously produces
\[
 |M|\ge\nu_k(G)/2,\qquad
 \Delta(M)\le O\left(k+\frac{\log^2(n/\rho)}{\eps}\right),\qquad
 \nu_1(M)\ge\nu_1(G)/342.
\]
The degree bound permits constant-factor expansion in the benchmark capacity $k$; the proof gives the explicit matching-preservation factor $342$.
\end{theorem}

Each phase randomly assigns active vertices to proposer and receiver roles. Proposers choose public candidate sets; receivers privately filter their incoming proposals with fresh masks. A pair enters $C$ only if both stages select it. The \emph{residual graph} consists of omitted true edges between still-active vertices. The analysis shows that a suitable count of these edges contracts by a factor at most $7/8$ in expectation per phase, giving logarithmically many phases. The output capacity can exceed the comparison capacity $k$ by a constant factor, unlike the main same-capacity result. \Cref{app:parallel} gives the algorithm and proof, including the matching-preservation factor $342$.

\subsection{Streaming edges in arbitrary order}
An insertion-only stream adds previously absent edges or a null record meaning ``no change.'' Let $T$ be its public length, $G_t$ the graph observed through position $t$, and $M_t$ the selected graph then. Neighboring streams replace one insertion by a null record. The stream is fixed independently of the algorithm's randomness. \emph{Continual privacy} protects the entire public output sequence.

A curator privately monitors $\nu_k(G_t)$, the best capacity-$k$ edge count, and rebuilds the shortlists after sufficiently large growth. The approximation factor becomes $2+\theta$, where $0<\theta\le1$ controls accuracy, with an additive error from monitoring. \Cref{thm:continual} gives the exact degree and error bounds, which hold together at every timestamp. Between rebuilds, participants decode from neighbor lists saved at the last refresh (or from edge arrival times). This keeps the selected edges and their degrees fixed. Intersecting the same candidates with a growing current list could instead create new selected edges.

\subsection{Streaming complete neighbor lists}
A more structured stream supplies each vertex's entire neighbor list in a public block of $n-1$ slots. Missing edges occupy null slots, so block lengths reveal no private degree; the two records for each edge agree. The curator runs one geometric sequential iteration per completed block. The outputs are therefore prefixes of one private interaction.

At a completed-block boundary, the selected degree is at most $b$ and $|M_t|\ge\nu_k(G_t)/2$, with the geometric capacity specified in \eqref{eq:seq-cap}. During an incomplete block, all new edges touch one vertex and can increase the comparison optimum by at most $k$, giving the temporary bound $|M_t|\ge\nu_k(G_t)/2-k/2$. These statements hold jointly over all positions with probability at least $1-\rho$ (\Cref{thm:blocks}). Because each released candidate set has a fully observed proposer, current observed neighbor lists suffice for decoding.

\paragraph{Resource comparison.}
\Cref{tab:resources} compares online communication, stored history, threshold tests, and decoder access in a common parameter regime. \Cref{app:resources} gives the full bounds. Threshold tests count evaluations, not arithmetic operations or running time.
\begin{table}[t]
\centering
\small
\setlength{\tabcolsep}{3pt}
\renewcommand{\arraystretch}{1.16}
\begin{tabularx}{\textwidth}{@{}>{\raggedright\arraybackslash}p{0.16\linewidth}>{\raggedright\arraybackslash}p{0.16\linewidth}>{\raggedright\arraybackslash}p{0.25\linewidth}>{\raggedright\arraybackslash}p{0.19\linewidth}>{\raggedright\arraybackslash}X@{}}
\toprule
Protocol; privacy & Rounds or refreshes & Public bits:\newline online / stored & Threshold tests & Private list used to decode \\
\midrule
Unit-step; local & $n$ rounds & $n^2$ / $n\log n$ & $n^2$ & Original list \\
Geometric; local & $n$ rounds & $n^2\log\log n$ / same & $n^2$ & Original list \\
Parallel; local & $\log n$ rounds & $n^2\log n\log\log n$ / same & $n\log^2 n$ & Original list \\
Arbitrary order; central & $\log n$ refreshes & $T+n^2\log n\log\log n$ / $n^2\log n\log\log n$ & $T+n^2\log n$ & Saved refresh-time list \\
Neighbor blocks; central & $n$ boundaries & $T+n^2\log\log n$ / $n^2\log\log n$ & $n^2$ & Current observed list \\
\bottomrule
\end{tabularx}
\caption{Resource comparison in a common parameter regime. All costs are $O(\cdot)$ upper bounds. Here $n$ counts vertices and $T$ stream positions. Accuracy and mask spacing are fixed, $\rho=n^{-c}$ for fixed $c>0$, and $T$ is polynomial in $n$. Online bits include public mask labels, one shared copy of each message, and repeat flags for unchanged streaming outputs. \Cref{tab:resources-exact} gives the full parameter-dependent bounds.}
\label{tab:resources}
\end{table}
\section{Discussion}\label{sec:discussion}
\paragraph{How accuracy affects capacity.}
The analysis certifies saturation at $k=b-D_4$, leaving an additive gap $D_4$ below capacity. Its guaranteed utility fraction is
\[
 \frac{b-D_4}{2(b+1)}=\frac12-\frac{D_4+1}{2(b+1)}.
\]
To make the approximation factor close to two, this loss must be small relative to $b$. Solving that requirement gives the $1/\eta$ dependence in \eqref{eq:main-cap}. The lower bound explains the dependence on privacy and failure probability for fixed accuracy; whether every method needs the same dependence on $\eta$ remains open. The factor-two example in \Cref{app:counterexamples} limits what saturation alone can prove.

\paragraph{Open directions.}
The sufficient capacities are worst-case guarantees, not evidence of performance for small directories. Structured graph families and experiments could clarify that setting. Another question is whether online communication can approach the compact stored record while preserving the guarantees. Weighted relationships, deletions, and privately choosing a particular matching are further directions for extending this form of coordination.
\label{main:lastpage}

\section*{Acknowledgments}
Felix Zhou acknowledges the support of the Natural Sciences and Engineering Research Council of Canada (NSERC).
Quanquan C. Liu and Felix Zhou are supported by a Google Academic Research Award and NSF Grant \#CCF-2453323.

\label{backmatter:statements}
\section*{Reproducibility Statement}
All protocols have explicit privacy budgets and finite stopping rules. \Cref{app:tools,app:lower,app:sequential,app:parallel,app:continual} contain the proofs; \Cref{app:compact,app:unit-comparison,app:representation} detail transcript reconstruction, capacity rounding, and the representation argument. \Cref{app:resources} derives the resource bounds and specifies the illustrative workflow. Public masks are sampled over potential pairs, and randomness is independent except where explicitly nested. \Cref{app:notation} provides a topic-organized notation guide.

\section*{AI Use Statement}
An AI assistant was used for literature and policy searches; mathematical formulation and proof development for the unit-step variant, threshold refinement, adjacent-mask analysis, transcript encoding, and representation theorem; proof auditing; implementation and execution of supplementary checks; and manuscript writing, editing, formatting, resource accounting, the illustrative workflow, and the explanatory rewrite and notation guide. The authors are responsible for the correctness, originality, attribution, and final content of this paper.


\label{backmatter:references}

\clearpage
\appendix
\label{backmatter:appendices}
\section{Privacy Tools, Encoding, and Combinatorial Facts}\label{app:tools}
Throughout the proofs, the graph and public parameters are fixed. Each public history determines the query functions and eligible sets; evaluating those functions uses the private adjacency lists. The sensitivity bounds hold for every history.

\paragraph{Probability notation.}
We use $\R$ for the real numbers. The law $\Lap(t)$ has density $e^{-|x|/t}/(2t)$ at $x\in\R$, where $t>0$ is its scale. A Bernoulli random variable, written $\mathrm{Ber}(p)$, is one with probability $p$ and zero otherwise. For a nonnegative integer $\ell$, the distribution $\mathrm{Bin}(\ell,p)$ is the sum of $\ell$ independent such variables. We write $\ind_A$ or $\ind\{A\}$ for the indicator of event $A$, equal to one when $A$ holds and zero otherwise. A superscript $c$ on an event denotes its complement. A random tape means a collection of random choices sampled in advance; fixing it makes those choices deterministic.

\subsection{Single-hit threshold selection}
In this subsection $D$ denotes the private dataset queried by the mechanism, and $D'$ a neighboring dataset; in a local invocation these are neighbor lists differing in one entry. The letter $T$ denotes a threshold here, not the streaming horizon used later. A query $q_j$ is sensitivity one when $|q_j(D)-q_j(D')|\le1$. The noisy threshold is sampled once and reused; query noise is sampled afresh each time.
For budget $a>0$ and a public threshold $T$, initialize $\widehat T=T+Z$, where $Z\sim\Lap(2/a)$. For each sensitivity-one query $q_j$, draw a fresh $Y_j\sim\Lap(4/a)$, answer whether $q_j(D)+Y_j\ge\widehat T$, and stop at the first positive answer. There is a deterministic finite bound on the number of queries, although the interaction may stop earlier. Queries may depend on previous public responses.

\begin{lemma}[Single-hit AboveThreshold]\label[lemma]{lem:AT}
The preceding interaction is $a$-DP. Reversing the signs of the queries and threshold gives a below-threshold test with the same guarantee. Publishing only the first positive index, with a distinguished default for an all-negative run, is also $a$-DP.
\end{lemma}
\begin{proof}
Fix the analyst's independent random tape and a transcript consisting of $j-1$ negative answers followed by a positive answer. Along this transcript, the query functions are fixed and satisfy $|q_i(D)-q_i(D')|\le1$ for neighboring inputs. Let $f$ be the density of the noisy threshold and let $F$ be the cumulative distribution function (cdf) of $\Lap(4/a)$, so $F(x)$ is the probability that its noise is at most $x$. The transcript probability on $D$ is
\[
 \int_{\R} f(z)\prod_{i<j}F(z-q_i(D))\,[1-F(z-q_j(D))]\,dz.
\]
Compare with the integrand on $D'$ at threshold value $z+1$. Since $q_i(D')\le q_i(D)+1$, all negative factors satisfy
$F(z-q_i(D))\le F(z+1-q_i(D'))$.
The threshold density shift costs at most $e^{a/2}$. The two positive-tail arguments differ by
$1+q_j(D)-q_j(D')\in[0,2]$.
A shift of at most two in a Laplace tail of scale $4/a$ costs at most $e^{a/2}$, by integrating the pointwise density-ratio bound. The product therefore costs at most $e^a$; a change of integration variable proves the claim for this transcript. An all-negative transcript has no positive factor and obeys the same bound. Summing transcripts, reversing $D,D'$, and averaging the analyst's randomness proves the DP inequality. The other output formats are postprocessing.
\end{proof}

\begin{lemma}[Deterministic threshold accuracy]\label[lemma]{lem:AT-accuracy}
If the threshold noise and all query noises have magnitude at most $h$, a positive answer implies $q_j(D)\ge T-2h$, and any reached query with $q_j(D)\ge T+2h$ is positive. For PVSM with $q_r=-(x+Z_r)$ and $T=-(b-2h)$, a selected index reached by a positive answer obeys $x+Z_r\le b$, and any reached index with $x+Z_r\le b-4h$ passes.
\end{lemma}
\begin{proof}
A positive answer implies $q_j(D)\ge T+Z-Y_j\ge T-2h$. If $q_j(D)\ge T+2h$, then $q_j(D)+Y_j\ge T+h\ge T+Z$. Substitute the PVSM queries and threshold. The non-strict test includes equality.
\end{proof}

\subsection{Finite public-mask representation}\label{app:encoding}
A public mask is a candidate subset, not a revealed subset of true neighbors. Its nested memberships can be described by one integer label per potential pair. In this subsection $D$ is that random integer label, rather than the private dataset notation used in the threshold proof.
Set $p_r=(1+s)^{-r}$ for $0\le r\le J$ and $p_{J+1}=0$. For each potential pair in a fresh family, independently sample a label $D\in\{0,\ldots,J\}$ with
\[
 \Prb[D=j]=p_j-p_{j+1},\qquad
 A_r=\{u\in S:D_{vu}\ge r\}.
\]
The probabilities telescope to one and give $\Prb[D\ge r]=p_r$. The analytical coupling
$D=\max\{r\le J:U\le p_r\}$ for $U\sim\mathrm{Unif}[0,1]$ gives the same masks. Thus every public label takes at most $\ceil{\log_2(J+1)}$ bits. Dyadic masks use the identical construction with $p_r=2^{-r}$.

Allocate potential-pair labels on an input-independent random tape and reveal them according to public eligibility and scheduling. Sampling over potential pairs keeps label generation independent of the private edge set. The finite label distribution is exact for the stated parameters. As with the Laplace mechanisms, finite-precision implementations require their own sampling and privacy analysis.

For geometric and dyadic schedules, PVSM tests exactly the indices $0,\ldots,J,J+1$, with $A_{J+1}=\varnothing$. If no positive answer occurs, it returns $J+1$. Returning the same index for an empty positive answer and for an all-negative run is postprocessing of the full threshold interaction. The unit-step schedule instead uses its $n$ positive indices and a distinct $\bot$ default as specified in \Cref{sec:unit-step}. Both conventions define a terminating output on every noise realization using a single invocation.

\subsection{Endpoint-wise accounting for interleaved thresholds}
The per-vertex stopping pattern follows the Multidimensional AboveThreshold approach of \citet{dhulipala2025}. Several local threshold mechanisms can be active at once, with messages from one determining later queries to another. The following argument handles that interleaving by factoring the probability of a complete response history. The number $a_v$ is the total privacy budget assigned to vertex $v$, while a mechanism slot is one potential invocation with its own independent noises.
\begin{lemma}[Endpoint privacy accounting]\label[lemma]{lem:endpoints}
Consider an interactive protocol in which the only private-row operations at vertex $v$ are independently randomized single-hit threshold mechanisms with deterministic total budget at most $a_v$. Their query descriptions, thresholds, and scheduling depend only on the public history; each query is sensitivity one in that row for every such history. All other public operations use only that history and an input-independent random tape. For graphs differing in edge $\{u,v\}$, the complete public transcript satisfies $(a_u+a_v)$-DP.
\end{lemma}
\begin{proof}
Fix the input-independent public tape and a complete finite response transcript $\tau$. Refine a PVSM index into its sequence of negative and positive threshold responses, without revealing any count or private noise. Preallocate every potentially used mechanism slot; an unused slot contributes probability factor one. Along $\tau$, each query function and public threshold is fixed.

Partition the independent private noises by mechanism slot $\ell$. The event that the execution follows $\tau$ is the intersection of the response constraints in those slots: every scheduling decision is determined by the fixed tape and preceding responses. Hence
\begin{equation}\label{eq:factorization}
 \Prb_G[\tau\mid\text{public tape}]=\prod_\ell L_{\ell,G}(\tau),
\end{equation}
where $L_{\ell,G}$ is the integral of that slot's threshold density, negative-answer cdf factors, and at most one positive-answer tail factor. Each factor integrates the response constraints for one independently randomized mechanism slot in the joint transcript event.

The change-of-variable proof of \Cref{lem:AT} applies to each affected slot for its fixed sequence of sensitivity-one functions, giving
$L_{\ell,G}(\tau)\le e^{a_\ell}L_{\ell,G'}(\tau)$.
The factors for all rows other than $u,v$ are equal, because their inputs are unchanged. Multiplying \eqref{eq:factorization} gives a likelihood ratio at most $e^{a_u+a_v}$. Sum over transcripts, average the public tape, and reverse the neighboring inputs. Published stopping indices and partially revealed public tapes are postprocessing of this refined experiment.
\end{proof}

The lemma is stated for the threshold protocols needed here. More general independently randomized interactive mechanisms can be treated by pure concurrent composition \citep{vadhanwang2021,vadhanzhang2023}. The slotwise factorization above gives a direct proof for the mechanisms used in this paper.

\paragraph{Sensitivity of the actual local queries.}
For a fixed history, both $C(v)$ and the eligible set $S$ are public. Saturation queries count $|N_G(v)\cap C(v)|$. Up to sign, a PVSM query is
\[
 |N_G(v)\cap C(v)|+|N_G(v)\cap A_r|,
 \qquad A_r\subseteq S,\quad S\cap C(v)=\varnothing.
\]
This is one count on the public union $C(v)\cup A_r$, with sensitivity one. Disjointness holds at every sequential invocation by the public vertex order, and at every parallel invocation because old candidates are excluded.

The following is the edge-level form of the Billboard Lemma of \citet[Lemma~9]{hsu2016}: the two changed endpoints are excluded, and all other local outputs are postprocessing.

\begin{lemma}[Billboard privacy away from changed endpoints]\label[lemma]{lem:billboard}
Let $B(G)$ be $(\eps,\delta)$-edge DP, and let each vertex apply a fixed measurable decoder to its adjacency list and $B$. For a fixed neighboring pair $G,G'$ differing in $\{u,v\}$, the transcript together with the joint decoded outputs of vertices outside $\{u,v\}$ satisfies the same DP inequality.
\end{lemma}
\begin{proof}
The rows supplied to those decoders are identical in the two experiments. Consequently one common measurable map takes $B$ to itself and all the stated outputs. Apply the DP inequality to the preimage of any output event. The unchanged rows make this a single postprocessing map for the neighboring inputs.
\end{proof}

\paragraph{Privacy accounting with public masks.}
Here the sampling indicator itself is part of the observed output. The example explains why we account for the full threshold budget even when a mask keeps a small proportion of candidates.
For a bit $x$, suppose $I\sim\mathrm{Ber}(p)$ is public and the release is $(I,Ix+Z)$ with $Z\sim\Lap(1/a)$. For $p>0$ and $t\ge1$, the likelihood ratio of the event $\{I=1,Ix+Z\ge t\}$ between $x=1$ and $x=0$ is exactly $e^a$. The public sampling indicator therefore prevents a universal smaller pure-DP parameter for this mechanism. Our accounting charges the full threshold budget, independent of a mask's sampling probability.

\subsection{Probability bounds used in the extensions}\label{app:tails}
These tail bounds control the probability of a large deviation from a random variable's mean. In the formulas below, $\mu$ is the expected count and $z$ is a relative deviation. In particular, a bound that holds for every mask before selecting an index can later be applied to the chosen mask.
For $Z\sim\Lap(t)$, $\Prb[|Z|>tL]=e^{-L}$. For a sum $X$ of independent Bernoulli variables of mean $\mu$, we use
\begin{equation}\label{eq:chernoff}
 \Prb[X\ge(1+z)\mu]\le e^{-z^2\mu/(2+z)},\qquad
 \Prb[X\le(1-z)\mu]\le e^{-z^2\mu/2},
\end{equation}
for $z>0$ in the first inequality and $0<z<1$ in the second. The lower-tail form appears in \citet[Theorem~1.1]{dubhashi2009}. For the upper tail, the exponential-moment calculation in their Section~1.4, Eq.~(1.5), gives the bound $\exp\{-\mu[(1+z)\log(1+z)-z]\}$; the inequality $(1+z)\log(1+z)-z\ge z^2/(2+z)$ yields the stated form for every $z>0$.

Each use of independence is conditional on the history preceding the fresh masks or roles. Tail bounds are then combined over the finite query horizon. The sequential proof uses \Cref{lem:thinning} and the Laplace tail.

\subsection{Combinatorial implications of saturation}\label{app:combinatorial}
We prove the three assertions of \Cref{lem:structural} separately. Throughout, $H\subseteq G$ means a subgraph on the same vertex set with fewer edges. The difference $G\setminus H$ denotes the omitted edges $E(G)\setminus E(H)$. The first argument bounds edge count, the second compares different capacity benchmarks, and the third preserves ordinary matchings.

\begin{lemma}[Residual-capacity charging]\label[lemma]{lem:charging}
If every edge of $G\setminus H$ has an endpoint of $H$-degree at least an integer $k\ge1$, then $|E(H)|\ge\nu_k(G)/2$, without an upper bound on $\Delta(H)$.
\end{lemma}
\begin{proof}
Fix an optimal $k$-matching $F$. We separate edges shared by $F$ and $H$ from edges present in only one, then charge each missing comparison edge to a saturated endpoint. Put $J=E(H)\cap F$, $H_0=E(H)\setminus F$, $F_0=F\setminus E(H)$, and $k_v=k-d_J(v)\ge0$. Every edge of $F_0$ has an endpoint with $d_{H_0}(v)\ge k_v$. Assign each edge to one such endpoint. A vertex with any assignment has $d_{H_0}(v)\ge k_v$ and receives at most $d_{F_0}(v)\le k_v$ edges; a vertex with no assignment receives zero. Thus
$|F_0|\le\sum_vd_{H_0}(v)=2|H_0|$.
It follows that $|F|=|J|+|F_0|\le|J|+2|H_0|\le2|E(H)|$.
\end{proof}

\begin{lemma}[Integral capacity comparison]\label[lemma]{lem:capacity}
For a simple graph and integers $1\le k\le b$, $\nu_k(G)\ge k\nu_b(G)/(b+1)$.
\end{lemma}
\begin{proof}
An optimal $b$-matching is a simple graph of maximum degree at most $b$. Vizing's edge-coloring theorem partitions it into at most $b+1$ ordinary matchings \citep[Section~28.1, Theorem~28.1]{schrijver2003}. Add empty classes if needed. The union of the $k$ largest classes has maximum degree at most $k$ and at least $k/(b+1)$ of the edges.
\end{proof}
The denominator $b+1$ accounts for integral feasibility: a triangle has $\nu_2=3$ and $\nu_1=1$, illustrating why a direct $k/b$ rescaling is insufficient.

The kernel framework was introduced by \citet{bhattacharya2015}. The preservation bound below is due to \citet[Lemma~2.4]{arar2018}.

\begin{lemma}[Matching preservation by a saturated kernel]\label[lemma]{lem:kernel}
Let $H\subseteq G$, let $1\le k\le b$ be integers, and suppose $\Delta(H)\le b$. If every edge of $G\setminus H$ has an endpoint of $H$-degree at least $k$, then
$\nu_1(H)\ge k\nu_1(G)/(2(b+1))$.
\end{lemma}
\begin{proof}
This is the edge-coloring argument of \citet[Lemma~2.4]{arar2018}, written with integer multiplicities. Fix a maximum ordinary matching $Q$ of $G$ and let $F=E(H)\setminus Q$. A multigraph may have several copies of the same pair. Form a multigraph $D$ by taking every edge of $F$ once and, for each $uv\in Q\cap E(H)$, adding
\[
 t_{uv}=\max\{0,b-d_F(u)-d_F(v)\}
\]
parallel copies of $uv$. If a vertex is not incident to $Q\cap E(H)$, its degree in $D$ equals its degree in $H$. At an endpoint $u$ of an edge $uv\in Q\cap E(H)$, its degree is $d_F(u)+t_{uv}\le b$. Hence $\Delta(D)\le b$.

For $uv\in Q\setminus E(H)$, neither endpoint belongs to another edge of $Q$, so its degrees in $D$ and $H$ agree. Saturation implies $d_D(u)+d_D(v)\ge k$. For $uv\in Q\cap E(H)$, put $a=d_F(u)+d_F(v)$; then
\[
 d_D(u)+d_D(v)=a+2\max\{0,b-a\}\ge b\ge k.
\]
The edges of $Q$ have disjoint endpoints, so $2|E(D)|\ge k|Q|$.

An edge-coloring assigns colors so that edges sharing an endpoint receive different colors; each color class is therefore an ordinary matching. To color $D$, first color the simple graph $F$ with $b+1$ colors. For $uv\in Q\cap E(H)$, at most $a$ colors appear at $u$ or $v$. Whenever $t_{uv}>0$, at least $b+1-a\ge t_{uv}$ colors are unused at both endpoints; assign distinct such colors to the parallel copies. Extensions for different edges of $Q$ do not conflict because their endpoints are disjoint. Each color class projects to an ordinary matching in $H$. Therefore
\[
 (b+1)\nu_1(H)\ge |E(D)|\ge k|Q|/2.
\]
This also covers an empty $Q$. The repeated edges lie on the matching $Q$, allowing the initial coloring to extend within the same $b+1$ colors.
\end{proof}
\section{Lower-Bound Proofs}\label{app:lower}
The three arguments use privacy to compare outputs on different input graphs. For inputs separated by $r$ edge additions or deletions, repeated application of the DP inequality gives a multiplicative factor $e^{r\eps}$ and, in approximate privacy, an additive allowance $\delta\sum_{j=0}^{r-1}e^{j\eps}$. This is the \emph{group-privacy} comparison. It requires privacy on the intermediate graphs, which is why the full graph domain matters.
\subsection{Fractional explicit overlap}\label{app:explicit-proof}
An explicit mechanism outputs weights $X_e\in[0,1]$ for all potential pairs $e\in\binom V2$. Its overlap with the true input is $U_G(X)=\sum_{e\in E(G)}X_e$. Total mass $\sum_e X_e$ counts both true edges and non-edges. The expected mass budget is $n\bar b/2$, corresponding to average expected degree $\bar b$. Here $\gamma\ge1$ is a multiplicative utility factor, $\beta\ge0$ an additive loss, and a perfect matching pairs all vertices. The symbol $b_v$ below denotes an individual expected degree budget, while $\delta_4$ accumulates approximate-privacy allowances across four edge edits.

\begin{theorem}[Explicit overlap under an expected mass budget]\label{thm:explicit}
Let $n\ge4$ be even. Suppose an $(\eps,\delta)$-edge-DP mechanism on all graphs satisfies $\E\sum_eX_e\le n\bar b/2$ on every input. If, for some $\gamma\ge1$ and $\beta\ge0$, its expected overlap on every perfect matching is at least $n/(2\gamma)-\beta$, then
\begin{equation}\label{eq:explicit}
 \frac{n}{2\gamma}-\beta
 \le\frac{e^{4\eps}n\bar b}{2(n-2)}+
 \frac n2\delta_4,
 \qquad \delta_4=\delta\sum_{j=0}^{3}e^{j\eps}.
\end{equation}
An explicit degree-$b$ bound gives $\bar b=b$. Expected heterogeneous degree budgets $b_v$ give $\bar b=n^{-1}\sum_vb_v$.
\end{theorem}

\begin{proof}[Proof of \Cref{thm:explicit}]
Let $\mathcal M$ be the family of perfect matchings on $V$. Symmetrize the mechanism by sampling a uniform permutation $\pi$, applying $A$ to $\pi G$, and relabeling the output by $\pi^{-1}$. For each fixed permutation this preserves neighboring inputs, and relabeling the output is postprocessing. Averaging over permutations preserves $(\eps,\delta)$-DP.

Set $U_* =\min_{G\in\mathcal M}\E U_G(A(G))$. On every perfect matching,
\[
 \E U_G(A_{\rm sym}(G))
 =\frac1{n!}\sum_\pi\E U_{\pi G}(A(\pi G))\ge U_*.
\]
The total expected weight remains at most $n\bar b/2$, because total weight is invariant under relabeling. The argument uses this relabeling-invariant total budget throughout.

The symmetrized output is equivariant under graph relabeling: relabeling the input has the same effect in distribution as relabeling the output. An automorphism is a relabeling that leaves the graph unchanged. Any non-edge of a perfect matching joins vertices from two different matching edges, and a graph-preserving relabeling can carry any such non-edge to any other. Thus each of the $n(n-2)/2$ non-edges has the same expected output weight $q_G$, and
\[
 \frac{n(n-2)}2q_G\le\frac{n\bar b}2,
 \qquad q_G\le\frac{\bar b}{n-2}.
\]
Fix $e=\{u,v\}\in E(G)$ and another matching edge $\{a,b\}$. Remove those two edges and add $\{u,b\},\{a,v\}$, obtaining another perfect matching $G'$. The edit distance is four, and $e$ is a non-edge of $G'$. Group privacy gives, for any event $S$,
\[
 \Prb[A_{\rm sym}(G)\in S]
 \le e^{4\eps}\Prb[A_{\rm sym}(G')\in S]+\delta_4.
\]
The privacy promise includes the four intermediate graphs. Since $0\le X_e\le1$, integrating the events $\{X_e>t\}$ over $0\le t\le1$ gives
\[
 \E[X_e\mid G]
 \le e^{4\eps}\E[X_e\mid G']+\delta_4
 \le e^{4\eps}\frac{\bar b}{n-2}+\delta_4.
\]
Summing over the $n/2$ true edges bounds the symmetrized overlap above by the right side of \eqref{eq:explicit}, while its overlap is at least $U_*\ge n/(2\gamma)-\beta$. Integral graphs are the special case $X_e\in\{0,1\}$. Expected degree bounds imply the expected total mass bound because the sum of weighted degrees counts each edge weight twice (the handshaking identity).
\end{proof}

\paragraph{Odd vertex counts.}
The even-$n$ family is enough for the asymptotic lower bound. An exact padded statement is also possible. For $n_{\mathrm{even}}=2\floor{n/2}\ge4$, fix $n_{\mathrm{even}}$ public vertices and isolate the rest. Under the original total mass budget $n\bar b/2$, symmetrizing only those $n_{\mathrm{even}}$ vertices gives
\[
 \frac{n_{\mathrm{even}}}{2\gamma}-\beta
 \le\frac{e^{4\eps}n\bar b}{2(n_{\mathrm{even}}-2)}+\frac{n_{\mathrm{even}}}{2}\delta_4.
\]
The numerator remains $n\bar b$, the total budget on the original vertex set. The four-edit comparison between perfect matchings applies within the padded domain as well.

\subsection{Probability-one degree feasibility}\label{app:zero-short-proof}
Here the candidate graph is public, but the degree hypothesis initially concerns its private intersection with the input. The complete-graph comparison transfers that hypothesis to the candidates themselves.
\begin{proof}
On the complete graph, $M=C(B)$, so the public event $\{\Delta(C(B))>b\}$ has probability zero. Every graph is at finite edge-edit distance from the complete graph. Pure group privacy transfers this zero probability to every input. Consequently $C(B)$ is an explicit degree-$b$ DP output. On a perfect matching its expected overlap is at least $(1-\rho)n/(2\gamma)$. Apply \Cref{thm:explicit} with $\delta=0$.
\end{proof}

\subsection{Intersection-based degree cost}\label{app:implicit-proof}
\begin{proof}[Proof of \Cref{thm:implicit-lower}]
Fix a potential center $r$. For a potential edge $e=\{r,v\}$, let $G_e$ contain only $e$, and let $G_\varnothing$ be empty. Since $\nu_1(G_e)=1$ and $\gamma$ is finite, success requires at least one decoded edge and hence requires $e\in C(B)$. Therefore
\[
 \Prb[e\notin C(B(G_e))]\le\rho.
\]
This event concerns only public candidates and is valid on the empty graph as well. Privacy gives
\[
 \Prb[e\notin C(B(G_\varnothing))]\le e^\eps\rho.
\]
The assumption on $\rho$ ensures $\log(1/\rho)\ge2\eps+\log16$, so the integer $q$ in the theorem is at least one and at most $n-1$. Choose any $q$ potential star edges. On the empty graph, their missing-candidate count $Z$ has expectation at most $q e^\eps\rho$. Markov's inequality gives
\[
 \Prb[Z\ge q/2\mid G_\varnothing]\le2e^\eps\rho.
\]
Now let $G_*$ be the star consisting of those edges, at edit distance $q$ from the empty graph. Pure group privacy, applied to the same candidate event, yields
\[
 \Prb[Z\ge q/2\mid G_*]
 \le2e^{(q+1)\eps}\rho
 \le2e^\eps\sqrt\rho\le\frac12.
\]
Outside this event, more than half of the $q$ true star edges are decoded, giving the degree assertion. The expectation and Markov bounds apply directly to the inclusion indicators of the chosen star.
\end{proof}

If the same decoder additionally has $\Delta(M)\le b$ with probability at least $1-\rho$, where $\rho<1/2$, it cannot have $b<q/2$: on the hard star, that would make the degree-cap failure probability at least $1/2$. The lower bound therefore applies to the joint utility/feasibility guarantees used in the main theorem. It also applies when utility is formulated as ordinary-matching preservation, because $|M|\ge\nu_1(M)$. The single-edge step requires positive utility on that input, as supplied by a finite multiplicative guarantee.

\subsection{Representation beyond declared intersection decoders}\label{app:representation}
A \emph{null set} of transcripts is an event with probability zero. Two transcript laws are \emph{mutually absolutely continuous} when they have the same null sets. Pure privacy gives this property across the graph domain. The proof first finds one set of probability one on which all decoders are valid and agree for every graph; it then compares their deterministic decisions there.
\begin{proof}[Proof of \Cref{thm:representation}]
Let $\mathcal G$ be the finite set of all simple graphs on $V$, and let $\mu_G$ be the law of $B(G)$ on its fixed measurable transcript space $(\mathcal B,\Sigma)$, where $\mathcal B$ is the set of possible transcripts and $\Sigma$ the collection of events whose probabilities are defined. Any two graphs are joined by a finite sequence of edge edits. Pure group privacy in both directions therefore gives mutual absolute continuity: $\mu_G(A)=0$ if and only if $\mu_{G'}(A)=0$ for every measurable $A$ and every $G,G'\in\mathcal G$. This null-set relation holds on the measurable transcript space, whether discrete or continuous.

For each graph $G$, let $F_G\subseteq\mathcal B$ consist of the transcripts for which all decoded sets are subsets of the corresponding neighborhoods and all endpoint decisions agree on $G$. Each $F_G$ is measurable: there are finitely many vertices and possible selected neighbors, and the decoders are measurable for each fixed neighborhood. The assumptions give $\mu_G(F_G)=1$. Mutual absolute continuity transfers this full-measure property to every $\mu_{G'}$. Since $\mathcal G$ is finite,
\[
 F_*:=\bigcap_{G\in\mathcal G}F_G
 \quad\text{satisfies}\quad \mu_{G'}(F_*)=1
 \quad\text{for every }G'\in\mathcal G.
\]
Thus for each fixed $B\in F_*$, validity and consistency hold \emph{simultaneously on all graphs}. This finite intersection is the step that licenses comparisons of decoder values on different inputs at the same transcript.

Fix $B\in F_*$ and distinct $u,v$. For any $S\subseteq V\setminus\{u\}$ containing $v$, consider the star graph with edges $\{\{u,w\}:w\in S\}$. Its neighborhoods at $u$ and $v$ are $S$ and $\{u\}$. Endpoint agreement on that graph implies
\[
 \ind\{v\in D_u(S,B)\}=\ind\{u\in D_v(\{u\},B)\}.
\]
Agreement on the single-edge graph $uv$ then gives
\begin{equation}\label{eq:singleton-representation}
 \ind\{v\in D_u(S,B)\}
 =\ind\{u\in D_v(\{u\},B)\}
 =\ind\{v\in D_u(\{v\},B)\}.
\end{equation}
In particular, the inclusion decision does not depend on any other neighbor of $u$, and the two singleton decisions are symmetric.

For every transcript, including those outside $F_*$, define a public undirected candidate graph by including $\{u,v\}$ exactly when $v\in D_u(\{v\},B)$ with $u$ the lower-ordered endpoint. This is a measurable function of the transcript, defined on the entire transcript space using only the fixed singleton decoders. For $B\in F_*$, \eqref{eq:singleton-representation} proves that every present edge is selected exactly when it lies in $C(B)$. Validity excludes all absent edges. Hence for every graph and vertex the decoded set is its neighborhood intersected with $C(B)$ throughout the common full-measure set. The common full-measure set gives the assertion for every input graph.
\end{proof}

\paragraph{Consequences for the two implicit barriers.}
Under \Cref{thm:representation}'s assumptions, replace the decoder by the measurable intersection representation without changing any input's output distribution except on a null set. If a degree cap $b$ holds almost surely on every graph, \Cref{prop:zero-failure} makes $C(B)$ an explicit degree-$b$ private output. If only a finite multiplicative cardinality guarantee is assumed with failure $\rho\le e^{-2\eps}/16$, \Cref{thm:implicit-lower} gives the same star degree obstruction, without assuming almost-sure degree feasibility. For an always-valid, consistently decoded ordinary matching, set $b=1$ in \Cref{prop:zero-failure}: a uniform constant-factor guarantee with bounded $\eps$ and fixed success probability is impossible for arbitrarily large $n$. Here the decoder's inputs and deterministic nature are exactly those specified in \Cref{thm:representation}.

\paragraph{Approximate privacy: a counterexample.}
The following construction shows the role of pure privacy in the representation theorem. Take $V=\{1,2,3\}$ and $0<\delta<1$. On the triangle and on the graph containing only edge $12$, release a symbol $\star$ with probability $\delta$ and otherwise release $0$. On every other graph release $0$. All output laws differ in total variation by at most $\delta$, so the mechanism is $(0,\delta)$-edge DP on the full domain. Let every decoder output the empty set on $0$, and on $\star$ output its entire neighborhood if it has size two, otherwise the empty set. These decoders are always valid and endpoint-consistent on every possible release: they select the triangle on the triangle and nothing on the single-edge graph. Both graphs produce $\star$ with positive probability, yet edge $12$ is selected in one and not in the other. No fixed $C(\star)$ can represent both intersections. Thus approximate privacy permits consistent local choices that have no single intersection representation.
\section{Complete Sequential Analysis}\label{app:sequential}
The proof has three parts: bound all relevant errors on one event, show that no selected degree exceeds capacity on that event, and show that every omitted edge touches a sufficiently high-degree vertex. The final graph inequalities then follow deterministically. Privacy is proved separately for the full distribution of public messages.

The two mask schedules use different public saturation thresholds. The unit-step schedule uses $b-2h_u$; the geometric schedule uses $b-6h$. Both use a negative PVSM threshold to select a candidate set whose true-neighbor count fits the remaining capacity.

\subsection{Unit-step implementation and accuracy}\label{app:unit-full}
The unit-step instantiation of \Cref{alg:sequential} uses $\tau_{\rm sat}=b-2h_u$. Its public parameters are \eqref{eq:unit-params}. At an active proposer's turn, sort the public eligible set by identifiers and remove its first $r$ elements to define $A_r$, for every $r=0,\ldots,n-1$. Once the set becomes empty it stays empty, but these indices remain distinct queries. The proposer initializes one fresh single-hit threshold with budget $a$ and public threshold $-(b-2h_u)$. It stops at the first positive test and publishes that index, or publishes $\bot$ if all $n$ tests are negative. The default selects the empty mask. Inactive proposers invoke no PVSM. Saturation answers, including negatives, are public; a positive saturation answer ends that vertex's persistent interaction permanently. The added threshold noises, noisy threshold values, and degree counts remain private; the base thresholds are public. The schedule has a deterministic finite horizon.

The local sensitivity argument is identical for either saturation threshold: fixing a history fixes $C(v)$ and $S$, and they are disjoint. Both local threshold mechanisms have budget $a=\eps/4$. \Cref{lem:endpoints} therefore bounds the public likelihood ratio for one changed edge by $e^{2(2a)}=e^\eps$, for the full output distribution. \Cref{lem:AT} includes the distinct outcomes $\bot$ and a positive empty-mask answer within the same privacy budget.

There are at most $n$ persistent thresholds, $n^2$ persistent query noises, $n$ proposer thresholds, and $n^2$ proposer query noises. They can all be preallocated independently. Since their scales are at most $16/\eps$, the probability that any used noise exceeds $h_u$ in magnitude is at most
\begin{equation}\label{eq:unit-failure}
 (2n^2+2n)e^{-L_u}
 =\rho\frac{n+1}{8(n+2)}<\rho.
\end{equation}
Call the complementary event $\mathcal E_u$. This event includes thresholds reused over many iterations. Preallocating all potential noises makes the bound uniform over the paths the protocol can take.

On $\mathcal E_u$, let $x$ denote a vertex's current selected degree. A negative saturation response satisfies $x<b$, and a positive response satisfies $x\ge b-4h_u\ge b-D_4$. The first inequality is strict because the threshold test uses $\ge$ for a positive answer. The receiver has integer old degree at most $b-1$ and gains at most one edge. A positive proposer response gives $x+Z_r\le b$; the default adds nothing, and inactive vertices receive nothing. Thus induction preserves the degree cap at every prefix, including the intermediate state after the saturation round.

For progress, let $x<k$ be the proposer's degree before selection and assume $b\ge k+D_4$. If a nonzero positive index $r$ is returned, rejection of $r-1$ and $Z_r\ge Z_{r-1}-1$ imply
\[
 x+Z_r>b-4h_u-1\ge b-D_4-1.
\]
Since $x+Z_r$, $b$, and $D_4$ are integers, the resulting degree satisfies $x+Z_r\ge b-D_4\ge k$. This includes equality $4h_u=D_4$; replacing the strict inequality by a weak one at this step would lose one unit. An empty mask has $Z_r=0$, so $x+Z_r=x<k\le b-D_4\le b-4h_u$ and passes whenever reached. At least the mask $A_{n-1}$ is empty. Hence a default is impossible for $x<k$. A positive empty-mask index $r>0$ is also incompatible with remaining below $k$ by the same progress argument. Index zero with an empty eligible set is allowed and selects every eligible neighbor vacuously. A proposer already at degree at least $k$ needs no progress; any of its default or empty outcomes is harmless for saturation.

After iteration $j$, take an omitted true edge with a processed endpoint and assume both current endpoint degrees are below $k$. Degrees are nondecreasing. Neither endpoint has ever deactivated, and the earlier endpoint has already proposed to an eligible set containing the later endpoint. Its degree after proposing was below $k$, so it returned index zero and selected all eligible pairs, a contradiction. The argument holds simultaneously at every prefix on $\mathcal E_u$. At the final prefix, residual-capacity charging and the known kernel conversion prove the two utility consequences of \Cref{prop:unit-step,thm:main}. Both utilities are deterministic consequences of the same event $\mathcal E_u$.

\subsection{From the invariant to the main guarantee}\label{app:main-conversion-proof}
The two approximation guarantees use the same saturation level. This proof also checks the simpler capacity bound in the main theorem.
\begin{proof}[Unit-step part of \Cref{thm:main}]
For $b<n-1$, set $k=b-D_4$. Condition \eqref{eq:main-cap} makes $1\le k\le b$ and gives exactly
\eqref{eq:unit-exact-cap}.
Apply \Cref{prop:unit-step,lem:structural}: both utilities follow on the same feasibility event. The complete graph handles $b\ge n-1$. For $H=\log(n/\rho)>\log8$, $L_u\le4H$ and $D_4\le256H/\eps+1$. Thus $((2+\eta)D_4+2)/\eta\le773H/(\eta\eps)<2^{10}H/(\eta\eps)$. Integer $b$ then satisfies the ceiling in \eqref{eq:main-cap}. The stored encoding follows below.
\end{proof}

\subsection{Exact capacity comparison within the unit-step family}\label{app:unit-comparison}
The alternative unit-step threshold $b-6h_u$ gives a second capacity calculation. It keeps the same PVSM and privacy budgets. Put $D_8=\ceil{8h_u}$. A positive saturation response gives degree at least $b-8h_u\ge b-D_8$; a negative gives degree below $b-4h_u<b$. For a proposer, the same rejected-predecessor calculation gives $x+Z_r>b-4h_u-1$. Because $h_u\ge1$, this exceeds $b-D_8-1$, so integrality gives degree at least $b-D_8$. The empty-mask and omitted-edge arguments above apply with every target $k\le b-D_8$. Thus this threshold certifies the smaller saturation target $b-D_8$. Its final capacity can be evaluated directly from the utility coefficient.

For either $D=D_8$ (threshold $b-6h_u$) or $D=D_4$ (threshold $b-2h_u$), choose $k=b-D$ in the nontrivial regime $b<n-1$. Both deterministic utility coefficients are $k/[2(b+1)]$. For $0<\eta\le1$, exact rearrangement yields
\begin{equation}\label{eq:unit-capacity-family}
 \frac{b-D}{2(b+1)}\ge\frac1{2+\eta}
 \quad\Longleftrightarrow\quad
 b\ge\frac{(2+\eta)D+2}{\eta}.
\end{equation}
Taking a ceiling gives the smallest integer capacity satisfying this coefficient inequality, and hence a sufficient capacity for the utility guarantee. Indeed the right side minus $D$ equals $(2D+2)/\eta>0$, so $k\ge1$ automatically. Also $k\le b<n-1$. For threshold $b-6h_u$, $b\ge k+8h_u$ follows from $D_8\ge8h_u$. For threshold $b-2h_u$, $b=k+D_4$ is precisely the premise of \Cref{prop:unit-step}. For $b\ge n-1$, use the complete candidate graph instead.

For the illustration in the main text, high-precision evaluation gives
\[
 h_u=486.458077407493\ldots,\qquad D_8=3892,\qquad D_4=1946.
\]
The three sufficient integer capacities are respectively
\[
 \ceil{4(8h_u+2)/0.1}=155747,\quad
 \ceil{(2.1D_8+2)/0.1}=81752,\quad
 \ceil{(2.1D_4+2)/0.1}=40886.
\]
The first exceeds $n-1$; the other two are below it. The supplementary script recomputes these sufficient capacities at two decimal precisions and checks the rounding.

\subsection{Explicit encoder and decoder for the unit-step transcript}\label{app:compact}
We encode the protocol's \emph{logical public history}: its saturation answers and proposer outcomes, with fixed public $n$, vertex order, identifiers, capacity, and privacy parameters. The deterministic unit-step masks are reconstructed from the public order and the recovered eligible sets.

\paragraph{Encoder.}
Read the completed public history only. For every vertex $v$, record $d_v$ as the iteration of its first positive saturation response, or $n+1$ if none occurs; this is the never-deactivated symbol. Record $p_v=r$ if its sole proposal returns positive index $r\in\{0,\ldots,n-1\}$, $p_v=n$ if it returns $\bot$, and $p_v=n+1$ if its proposer turn is skipped. The alphabets have $n+1$ and $n+2$ symbols, respectively. Encode $d_v-1$ in $\ceil{\log_2(n+1)}$ bits and $p_v$ in $\ceil{\log_2(n+2)}$ bits, concatenated in vertex order. This uses exactly \eqref{eq:compact-bits}. The skipped symbol is redundant given $d_v$, but retaining it yields an explicit alphabet without identifying distinct outcomes.

\paragraph{Decoder.}
Parse the fixed-width fields. Starting from $C=\varnothing$ and all vertices active, replay iterations $i=1,\ldots,n$. Each vertex with $d_v\ge i$ sends its negative saturation answer when $d_v>i$ and its positive answer when $d_v=i$; earlier-deactivated vertices send nothing. Remove the positive vertices. The remaining active set is exactly $\{v:d_v>i\}$. If $d_{v_i}\le i$, check that $p_{v_i}=n+1$ and skip the proposer. Otherwise form the later active set, sort it by identifiers, and reconstruct all $n$ unit-step masks. For $p_{v_i}<n$, return that positive index and add its candidate star; for $p_{v_i}=n$, return the all-negative default and add nothing. All logical public messages and candidate sets are now determined. Unused binary codewords or inconsistent pairs of fields may be rejected; they are never produced by the encoder. This convention makes decoding a total measurable procedure on encoded strings.

\paragraph{Losslessness and boundary cases.}
At the initial state the original and replayed histories agree. If they agree before iteration $i$, the stored first-positive times reconstruct precisely the vertices queried and all their saturation answers, including negatives. Hence both executions have the same active set, skipped-proposer decision, and eligible set. The masks are deterministic functions of that set, so the stored proposal symbol reconstructs the outcome and the same candidate update. Induction proves agreement at every prefix. If the eligible set is empty, all $n$ masks are empty but positive indices $0,\ldots,n-1$ and the all-negative default remain distinct. The same applies to indices after a nonempty eligible set has been exhausted. No accuracy assumption enters this argument, so it works on every realization, including those violating the desired degree cap.

Encoding is a function of the completed public history and is therefore DP postprocessing. The inverse on realizable histories establishes losslessness. The direct implementation sends up to $n^2$ saturation answers and runs $n$ local PVSM invocations with up to $n$ tests each, using $O(n)$ synchronized rounds. Equation~\eqref{eq:compact-bits} bounds storage of that completed interaction; \Cref{app:resources} gives the separate online resource counts.

\subsection{Geometric parameters and comparison guarantee}\label{app:geometric-ledger}
For a proposer with public eligible set $S$, independently draw a public uniform number $U_{vu}\in[0,1]$ for each potential pair. Let $s\in(0,1/4]$ be the mask spacing, $A_r$ the candidates retained at index $r$, $Z_r$ their true-neighbor count, and $J$ the final random-mask index:
\begin{equation}\label{eq:masks}
 A_r=\{u\in S:U_{vu}\le(1+s)^{-r}\},\quad
 Z_r=|N_G(v)\cap A_r|,\quad J=\ceil{\log_{1+s}n}.
\end{equation}
Append the empty mask $A_{J+1}$ and test indices $0,\ldots,J+1$. A positive test at the empty mask and an all-negative run both return $J+1$. A finite integer label records the same nested memberships as each public uniform number (\Cref{app:encoding}).

For $s\in(0,1/4]$, use saturation threshold $\tau_{\rm sat}=b-6h$ and the following public parameters in \Cref{alg:sequential}:
\begin{equation}\label{eq:seq-params}
\begin{split}
 L&=\log\frac{32n(n+J+3)}\rho,\qquad a=\eps/4,\qquad h=16L/\eps,\\
 \lambda&=1-2s,\qquad E_0=8h+2L.
\end{split}
\end{equation}
\begin{theorem}[Geometric sequential saturation]\label{thm:sequential}
For integers $1\le k\le n-1$ and
\begin{equation}\label{eq:seq-cap}
 b\ge k/\lambda+E_0,
\end{equation}
\Cref{alg:sequential} is $\eps$-LEDP and uses $O(n)$ rounds. With probability at least $1-\rho$, every prefix has degree at most $b$ and satisfies the prefix saturation invariant in \Cref{prop:unit-step}. Consequently $|M|\ge\nu_k(G)/2$.
\end{theorem}

\subsection{Geometric implementation and privacy}
In the remainder of this appendix the masks are geometric, with $\tau_{\rm sat}=b-6h$. Saturation answers are public; a positive answer permanently deactivates a vertex. Each iteration has one simultaneous saturation round and at most one proposer invocation. All threshold tests internal to that PVSM are local to the proposer, so it publishes only a stopping index. Public mask generation can be combined with a constant number of synchronization steps per iteration. Thus the protocol uses $O(n)$ rounds even when $J$ is large.

Before $v_i$ proposes, every candidate pair already incident to $v_i$ has its other endpoint earlier in the order. Its eligible set of later active vertices is therefore disjoint from $C(v_i)$. Each local query has sensitivity one by \Cref{app:tools}. A vertex has one persistent mechanism and at most one PVSM, each of budget $a=\eps/4$. \Cref{lem:endpoints} gives total privacy loss $2(2a)=\eps$. This accounting covers the full public transcript. The final intersection is computed privately at each vertex.

\subsection{Adjacent-mask concentration}\label{app:thinning-proof}
The random variable $Z_r$ counts true neighbors in mask $r$. Rather than estimate adjacent counts separately, we bound their difference: $D_r$ is the number of neighbors removed between two masks. This is distinct from the fixed unit-step allowance $D_4$. Conditional on the earlier mask, those removals are independent Bernoulli trials.
\begin{lemma}[Adjacent-mask thinning]\label[lemma]{lem:thinning}
Fix the history before a fresh mask family is drawn. For $0<s\le1/4$ and $L>0$, with probability at least $1-(J+1)e^{-L}$,
\begin{equation}\label{eq:thinning}
 Z_r\ge(1-2s)Z_{r-1}-L\quad(1\le r\le J),\qquad Z_J\le L+2.
\end{equation}
The bound holds jointly across the nested scales.
\end{lemma}
\begin{proof}[Proof of \Cref{lem:thinning}]
At a fixed transition, condition on the true-neighbor set in $A_{r-1}$, of size $t$. The deletions $D_r=Z_{r-1}-Z_r$ are conditionally $\mathrm{Bin}(t,s/(1+s))$. Thus $\E[e^{D_r}\mid A_{r-1}\cap N_G(v)]\le e^{2st}$, since $e-1<2$. Exponential Markov gives $\Prb[D_r>2st+L\mid A_{r-1}\cap N_G(v)]\le e^{-L}$. Also $\E Z_J\le1$ implies $\E e^{Z_J}<e^2$ and $\Prb[Z_J>L+2]\le e^{-L}$. Averaging and union-bounding over the $J$ transitions and terminal count proves the claim. The bounds are established jointly over all indices before selection.
\end{proof}

\subsection{The simultaneous accuracy event}
The event $\mathcal E$ defined below requires all private noises to be small and all random masks to satisfy the needed count bounds. Bounding its failure probability uniformly over past histories allows the algorithm to select masks based on previous outcomes.
All private noise scales are at most $4/a$, so each sampled noise has magnitude at most $h=4L/a$ except with probability $e^{-L}$. There are at most $n^2$ persistent query noises, $n$ persistent thresholds, and $n(J+3)$ noises from PVSM invocations (one threshold and $J+2$ queries per invocation). The total is at most $n^2+n(J+4)$. Unused noises may be preallocated independently.

For each proposer, condition on the full execution history before its fresh mask family is drawn. Its eligible true-neighbor set is then fixed. \Cref{lem:thinning} fails with conditional probability at most $(J+1)e^{-L}$. A union bound over at most $n$ invocations gives an event $\mathcal E$ on which all noises are bounded and all adjacent-mask and terminal bounds hold. Its failure probability is at most
\begin{equation}\label{eq:seq-failure}
 \Prb[\mathcal E^c]
 \le n(n+2J+5)e^{-L}
 =\rho\frac{n+2J+5}{32(n+J+3)}<\rho.
\end{equation}
For each invocation, the conditional tail bound holds for every preceding history. Averaging over histories and taking a union bound establishes the common event $\mathcal E$.

\subsection{Feasibility at every prefix}
Work on $\mathcal E$. A positive saturation answer at threshold $b-6h$ implies selected degree at least $b-8h\ge k$. A vertex whose answer is negative has degree strictly below $b-4h$ immediately before the proposal. Selected degrees only increase, so the lower bound persists after deactivation.

There is one proposer per iteration. An active receiver consequently gains at most one true edge. Its prior degree is an integer strictly less than the integer cap $b$, so its new degree is at most $b$. A proposer adds pairs disjoint from its previous candidates; any index returned through a positive PVSM test has resulting degree $x+Z_r\le b$. The empty default changes no degree. Inactive vertices receive no further pairs. Induction proves degree at most $b$ after every iteration.

\subsection{Progress, terminal masks, and the prefix invariant}
Fix an active proposer $v$ with old degree $x=|N_G(v)\cap C(v)|<k$. Index zero selects every eligible neighbor. For a selected index $1\le r\le J$, rejection of the previous index and \Cref{lem:AT-accuracy} give $x+Z_{r-1}>b-4h$. The thinning inequality gives
\begin{align*}
 x+Z_r
 &\ge\lambda(x+Z_{r-1})+(1-\lambda)x-L\\
 &>\lambda(b-4h)-L\\
 &\ge k+4\lambda h+(2\lambda-1)L\ge k.
\end{align*}
Here $b\ge k/\lambda+8h+2L$, $x\ge0$, and $\lambda=1-2s\ge1/2$.

The empty index cannot be reached as the selected answer while $x<k$ on $\mathcal E$. Indeed, $Z_J\le L+2$, and
\[
 x+Z_J<k+L+2\le b-4h,
\]
because $b-4h\ge k/\lambda+4h+2L$, $h\ge1$, and $L>0$. Thus index $J$ passes if reached. The proof includes index $r=1$, the possibility of no eligible true neighbor, and equality in the capacity test. If the proposer already has $x\ge k$, no progress assertion is needed.

After iteration $j$, suppose an omitted true edge has a processed endpoint and both its current selected endpoint degrees are below $k$. Let $v_i$ be its earlier endpoint in the public order, so $i\le j$. Neither endpoint can previously have deactivated, because that would imply degree at least $k$ thereafter. At iteration $i$, the later endpoint was eligible. Since the proposer's degree is still below $k$, its proposal could not have selected a nonzero index; the progress argument forces index zero, which selects that edge. This is a contradiction. The stated invariant holds for every prefix on the one event $\mathcal E$.

At $j=n$, every omitted edge is covered by the invariant. \Cref{lem:charging} gives $|M|\ge\nu_k(G)/2$. Together with \eqref{eq:seq-failure}, feasibility, privacy, and the round count, this completes \Cref{thm:sequential}. For constant $s$, $J=O(\log n)$ and $E_0=O(\log(n/\rho)/\eps)$, giving the benchmark-one degree bound.

\subsection{Geometric same-capacity reduction}\label{app:geometric-reduction}
\begin{proof}[Geometric-mask part of \Cref{thm:main}]
For $b<n-1$, set $s=\eta/16$, so $\lambda=1-\eta/8$, and take
\begin{equation}\label{eq:exact-cap}
 b\ge8(E_0+2)/\eta,\qquad k=\floor{\lambda(b-E_0)}.
\end{equation}
Then $k\ge1$, \eqref{eq:seq-cap} holds, and the floor gives
\begin{equation}\label{eq:rounding}
 \frac{k}{b+1}\ge1-\frac\eta8-\frac{E_0+2}{b+1}\ge1-\frac\eta4.
\end{equation}
Apply \Cref{lem:structural} as in the unit-step proof. \Cref{app:geometric-constant} verifies that \eqref{eq:geometric-main-cap} implies \eqref{eq:exact-cap}. The complete candidate graph covers $b\ge n-1$.
\end{proof}

\subsection{Joint utility and integer rounding}
For the geometric same-capacity guarantee in \Cref{thm:main}, choose $s$ as in \Cref{app:geometric-reduction} and use the same privacy allocation. The integer $k=\floor{\lambda(b-E_0)}$ is public. Under \eqref{eq:exact-cap}, $b-E_0$ is positive and large enough that $k\ge1$; also $k\le b$. Inequality \eqref{eq:rounding} follows from
\[
 k\ge\lambda(b-E_0)-1,
 \qquad
 \frac{k}{b+1}\ge\lambda-\frac{\lambda E_0+\lambda+1}{b+1}
 \ge1-\frac\eta8-\frac{E_0+2}{b+1}.
\]
The cardinality comparison and the matching-kernel lemma are deterministic implications of the same prefix-feasibility and saturation event. Consequently, the two utility claims in \eqref{eq:main-guarantees} hold jointly with the original success probability.

\subsection{Explicit constant for the geometric same-capacity bound}\label{app:geometric-constant}
Put $H=\log(n/(\eta\rho))>\log8$. With $s=\eta/16$, the inequality $\log(1+s)\ge s/(1+s)$ gives $J\le17\log n/\eta+1$ and $n+J+3\le20n/\eta$. Consequently
\[
 L\le\log\frac{640n^2}{\eta\rho}\le6H,\qquad
 E_0=(128/\eps+2)L\le780H/\eps.
\]
Since $H>1$ and $\eps\le1$, we obtain
\[
 8(E_0+2)/\eta\le6256H/(\eta\eps)<2^{13}H/(\eta\eps).
\]
Thus \eqref{eq:geometric-main-cap} implies the parameter-dependent sufficient condition \eqref{eq:exact-cap}, with the stated absolute constant. The resulting cap is the sufficient bound stated in \Cref{thm:main}.
\section{A Logarithmic-Round Protocol}\label{app:parallel}
We prove \Cref{thm:parallel} by leaving capacity slack between the deactivation threshold and the output cap. This slack supports progress in parallel, and the resulting saturation gives both the benchmark cardinality and ordinary-matching guarantees.

\subsection{Parameters and implementation}
Here $R$ is the number of phases, $J$ the last nonempty random-mask index, and $L$ the tail-bound parameter. The privacy budgets are $a_0$ for the persistent test and $a_1$ for a single proposal or acceptance call. The noise allowance $h$ also sets the saturation target and output cap. A dyadic probability is a power of one half, $2^{-r}$.
For the target capacity $k$, set
\begin{equation}\label{eq:par-params}
\begin{split}
 R&=\ceil{\frac{\log(2n^2/\rho)}{\log(8/7)}},\qquad
 J=\ceil{\log_2n},\qquad
 L=\log\frac{32nR(J+3)}\rho,\\
 a_0&=\eps/4,\qquad a_1=\eps/(4R),\qquad
 h=\max\{1,k/6,4L/a_1\},\qquad b=\ceil{1024h}.
\end{split}
\end{equation}
Each vertex has a persistent $\AT_{a_0}$ at threshold $8h$. In every phase, fresh independent role bits assign active vertices to proposer or receiver roles. Use separate fresh arrays of public pair labels for proposal and acceptance masks. The analytical couplings use independent uniforms on every potential unordered pair, across arrays and phases. A PVSM has budget $a_1$, threshold $-(b-2h)$, and dyadic inclusion probabilities $2^{-r}$ for $r=0,\ldots,J$, followed by the empty mask. Previously selected candidate pairs are excluded.

\begin{algorithm}[t]
\caption{$\Par$: proposal--acceptance with capacity slack}\label{alg:parallel}
\KwIn{Private adjacency lists and public parameters \eqref{eq:par-params}.}
Initialize $C\gets\varnothing$, active set $V_*\gets V$, and a persistent $\AT_{a_0}$ at threshold $8h$ at each vertex\;
\For{$t=1,\ldots,R$}{
Independently assign each $v\in V_*$ to public proposer set $P$ or receiver set $Q$, with probability $1/2$ each\;
\For{each $v\in P$, in parallel}{
Run PVSM on $Q\setminus C(v)$ with budget $a_1$ and fresh dyadic proposal masks\;
Publish its index, defining a public proposal set $P_v$\;
}
\For{each $w\in Q$, in parallel}{
Set $S_w=\{v\in P:w\in P_v,\ \{v,w\}\notin C\}$\;
Run PVSM on $S_w$ with budget $a_1$ and fresh independent dyadic acceptance masks\;
Publish its index, defining a public acceptance set $Q_w$\;
}
Add to $C$ every pair $\{v,w\}$ with $v\in P$, $w\in P_v$, and $v\in Q_w$\;
Every active vertex queries its persistent mechanism on $|N_G(v)\cap C(v)|$; positive answers permanently remove it from $V_*$\;
}
Each vertex privately decodes $N_G(v)\cap C(v)$\;
\end{algorithm}

All role bits, masks, indices, and saturation answers are public. A coordinator need not know whether a proposal is a true edge; both endpoints check public mutual selection and intersect with their private rows. The direct representation uses two finite-label pair arrays per phase, so including active flags and indices gives $O(Rn^2\log(J+2))$ public bits. The bit count includes the public randomness used to define both mask arrays.

\subsection{Privacy and feasibility}
A vertex participates in at most one PVSM per phase, because proposer and receiver roles are disjoint. It has at most $R$ such interactions of budget $a_1$ and one persistent interaction of budget $a_0$. All count queries have sensitivity one. \Cref{lem:endpoints} gives
\[
 2(a_0+Ra_1)=2(\eps/4+\eps/4)=\eps.
\]
There is a constant number of synchronization steps per phase. All loops are finite and the empty default applies on every noise realization.

Suppose all noises exposed so far are bounded by $h$. A positive saturation test implies degree at least $8h-2h=6h\ge k$. A vertex still active after the last test has degree below $8h+2h=10h$; initially its degree is zero. In a phase, a proposer offers at most $b-d_M(v)$ true edges by \Cref{lem:AT-accuracy}, and only a subset can be accepted. A receiver accepts at most $b-d_M(w)$ true edges. The old and newly selected pairs are disjoint. Thus every degree remains at most $b$ throughout the execution on the bounded-noise event.

\subsection{Independent reference masks}
The actual masks are selected using private counts and are therefore data-dependent. The analysis compares them with simpler \emph{reference masks}: masks at fixed indices once the pre-phase state is known. These are proof devices, not extra outputs. On the accuracy event, the selected masks contain the relevant reference masks.
At the start of a phase, define the true residual graph
\[
 H=(V_*,(E(G)\setminus M)[V_*]),\qquad d_H(v)=\deg_H(v),
\]
where $V_*$ is the active vertex set and $F[V_*]$ restricts an edge set $F$ to pairs with both endpoints in $V_*$. Thus $H$ consists of still-omitted true edges between active vertices. Set the proof parameter $K_0=384h$, which is unrelated to the refresh limit $K$ in the streaming construction. Given the pre-phase history, define
\[
 p_v=\begin{cases}
 1,&d_H(v)\le K_0,\\
 \text{largest dyadic }2^{-r}\le K_0/d_H(v),&d_H(v)>K_0.
 \end{cases}
\]
For the second case, $K_0/(2d_H(v))<p_v\le K_0/d_H(v)$ and its index belongs to $\{0,\ldots,J\}$. These reference probabilities are analytical quantities defined from the residual graph.

Condition on the pre-phase history and all role bits. The reference number $X$ of receiver neighbors sampled by a proposer has mean at most $K_0$. When $p_v=1$, its count is at most $K_0$ deterministically. In the other case, independence gives $\E[2^X]\le e^{\E X}\le e^{K_0}$, so
\begin{equation}\label{eq:reference-upper}
 \Prb[X>2K_0]\le e^{-(2\log2-1)K_0}\le e^{-K_0/3}.
\end{equation}
Here $\log2\ge2/3$. This bound also covers mean zero. Since an active vertex has old degree below $10h$,
\[
 10h+2K_0=778h\le b-4h.
\]
On the bounded-noise event, the reference index passes if reached. Nesting therefore implies that the actual selected proposal mask contains the reference mask on the reference-count event from \eqref{eq:reference-upper}. This containment lets the progress argument use the reference masks in place of the adaptively selected masks.

For a receiver, condition on the history, all roles, and all proposal-stage choices. Its acceptance uniforms are still fresh. Let $q$ be the actual number of eligible true proposals. If $q\le K_0$, index zero passes and all proposals are accepted. If $q>K_0$, select the largest dyadic reference probability at most $K_0/q$. Its count has mean $\mu\in(K_0/2,K_0]$. The event $X<K_0/4$ implies $X<\mu/2$, which by \eqref{eq:chernoff} has probability at most $e^{-K_0/16}$. Together with \eqref{eq:reference-upper}, except with probability at most $2e^{-K_0/16}$, its count is in $[K_0/4,2K_0]$. The upper bound makes that index pass; nesting forces the chosen mask to contain at least $K_0/4$ true proposals. Therefore a receiver either accepts all its proposals or accepts at least $K_0/4$ on this accuracy event.

Let $\mathcal E$ include bounded noises and all these proposal and acceptance reference bounds in every phase. There are at most $nR(J+5)$ noises, including the initial persistent thresholds. Each noise fails with probability at most $e^{-L}$, and the reference failures sum to at most $3nR e^{-K_0/16}$. Since $h\ge4L/a_1\ge16L$, we obtain
\begin{equation}\label{eq:par-failure}
 \Prb[\mathcal E^c]
 \le nR(J+8)e^{-L}
 =\rho\frac{J+8}{32(J+3)}\le\rho/2.
\end{equation}
Each bound is applied before drawing the relevant fresh array, even on histories where a previous accuracy condition failed.

\subsection{Progress on a constant fraction of residual edges}
We temporarily give each residual edge a direction for the analysis. The head is the endpoint it points toward and the tail the other endpoint. The incoming and outgoing degrees, written $\deg_H^-(v)$ and $\deg_H^+(v)$, count arrows entering and leaving $v$. This orientation does not change the undirected protocol.
Orient each edge of $H$ from lower to higher residual degree, breaking ties by public identifiers. Call a vertex good if its incoming degree is at least one third of its total degree, and call an edge good if its head is good. At least half the edges are good. Indeed, for the set $B$ of non-good vertices,
\[
 \sum_{v\in B}\deg_H^-(v)
 \le\tfrac12\sum_{v\in B}\deg_H^+(v)
 \le\tfrac12|E(H)|.
\]
The left side counts non-good edges. The inequality includes the empty graph. The counting fact holds for any orientation; the degree order is used in the probability comparison below.

Fix a good directed edge $(u,v)$ in a phase whose past accuracy conditions hold. All residual degrees and reference probabilities are fixed before the new roles and masks are drawn.

\paragraph{Low-degree head.}
If $d_H(v)\le K_0$, then $d_H(u)\le d_H(v)$ and $p_u=1$. The event that $u$ is a proposer and $v$ a receiver has probability $1/4$. On this event and current accuracy, $u$ proposes $v$, and $v$ has at most $K_0$ true proposals and accepts them all. Thus the edge leaves the residual graph. An accuracy failure will instead send the analysis potential to zero.

\paragraph{High-degree head.}
Suppose $d_H(v)>K_0$ and condition only on $v$ being a receiver. For each incoming neighbor $u$, use the event that $u$ is a proposer and the fresh proposal uniform on $\{u,v\}$ is at most $p_u$. These events are independent across $u$: the roles and pair uniforms are distinct, and all $p_u$ are pre-phase functions. Each event has probability at least
\[
 \tfrac12\min\{1,K_0/(2d_H(u))\}
 \ge\frac{K_0}{4d_H(v)}.
\]
There are at least $d_H(v)/3$ incoming neighbors, so the count $X$ of these reference proposals has mean at least $K_0/12=32h$. Chernoff gives
\[
 \Prb[X<16h\mid v\text{ is a receiver}]\le e^{-4h}.
\]
The remaining role bits and pair uniforms retain their independent distributions in this conditional calculation. If $X\ge16h$ and current accuracy holds, the actual proposal masks include those pairs. The receiver accepts all proposals when $q\le K_0$, or at least $K_0/4=96h$ when $q>K_0$. Its new degree is therefore at least $16h$, forcing a positive saturation test. All residual edges incident to $v$ disappear. The baseline event just used has probability at least $(1-e^{-4h})/2\ge1/4$.

\subsection{A stopped-potential contraction}
A \emph{potential} is a nonnegative quantity used to measure remaining work. Here it counts residual edges until an accuracy condition fails, after which it is set to zero. The analysis keeps the fresh randomness unconditioned, rather than first assuming every phase succeeds. A \emph{filtration} $\mathcal F_t$ is the information revealed for analysis through time $t$; conditional expectation given $\mathcal F_t$ averages only over later randomness.
Expose, for analysis only, all role bits, mask arrays, and private noises used through each completed phase. Let $\mathcal F_t$ be the resulting filtration, including the fixed input graph and the initial persistent threshold noises but not future roles, masks, or query noises. Let $\mathcal E_t$ be the event that all accuracy conditions exposed through phase $t$ hold, and let $H_t$ be the actual residual graph after that phase. Define
\[
 Z_t=\ind_{\mathcal E_t}|E(H_t)|.
\]
The events $\mathcal E_t$ decrease with $t$, so $Z_t$ stays zero after a failure, and $Z_0\le n^2$.

On a history in $\mathcal E_t$, every current good edge $e$ has a next-phase baseline event $B_e$ of conditional probability at least $1/4$, established above using fresh coins. On $B_e\cap\mathcal E_{t+1}$, the edge is absent from $H_{t+1}$; on $\mathcal E_{t+1}^c$, its stopped indicator is zero. Hence
\[
 \E[\ind_{\mathcal E_{t+1}}\ind_{\{e\in E(H_{t+1})\}}\mid\mathcal F_t]
 \le1-\Prb[B_e\mid\mathcal F_t]\le3/4.
\]
The indicator for another current edge is at most one. Residual edges are never added, because candidates only accumulate and the active set only shrinks. Since at least half the current edges are good, summing gives
\[
 \E[Z_{t+1}\mid\mathcal F_t]\le\tfrac78 Z_t.
\]
This also holds outside $\mathcal E_t$, where the potential remains zero. Iterating conditional expectations gives
\[
 \E Z_R\le(7/8)^R n^2\le\rho/2.
\]
By Markov's inequality, $\Prb[\mathcal E_R\cap\{E(H_R)\ne\varnothing\}]\le\rho/2$. Together with \eqref{eq:par-failure}, this gives $\Prb[\mathcal S]\ge1-\rho$ for the joint event $\mathcal S=\mathcal E_R\cap\{E(H_R)=\varnothing\}$. Accuracy alone does not imply that the residual graph is empty.

On $\mathcal S$, every omitted edge has a deactivated endpoint of selected degree at least $6h\ge k$. \Cref{lem:charging} gives $|M|\ge\nu_k(G)/2$. Since selected degrees are integers, the stronger saturation threshold $\kappa=\ceil{6h}\le b$ is valid. \Cref{lem:kernel} yields
\[
 \nu_1(M)\ge\frac{\kappa}{2(b+1)}\nu_1(G)
 \ge\frac{6h}{2(1024h+2)}\nu_1(G)
 \ge\frac1{342}\nu_1(G),
\]
where the last inequality uses $h\ge1$. This guarantee holds on the same joint success event $\mathcal S$, with no additional failure term.

Finally, $R=O(\log(n/\rho))$, $L=O(\log(n/\rho))$, and
\[
 b=\ceil{1024\max\{1,k/6,16RL/\eps\}}
 =O\left(k+\frac{\log^2(n/\rho)}\eps\right).
\]
Together with privacy and the constant synchronization cost per phase, this proves \Cref{thm:parallel}.
\section{Insertion-Only Continual Release}\label{app:continual}
Both continual constructions use the geometric sequential variant with saturation threshold $b-6h$ and the prefix guarantee of \Cref{thm:sequential}. Here $k$ is the comparison capacity. In the arbitrary-order theorem, $\theta\in(0,1]$ controls approximation, $s\in(0,1/4]$ is fixed geometric-mask spacing, and $\lambda=1-2s$ is the corresponding retention coefficient.

A trusted curator processes a stream fixed independently of the mechanism's randomness. Privacy protects the complete sequence of public outputs. This fixed-stream model is distinct from adaptive-stream privacy \citep{bun2026}.

\subsection{Arbitrary edge order and snapshot decoding}
A snapshot is the graph, or a vertex's neighbor list, saved at a specified stream position. A refresh runs a new static coordination protocol on that snapshot. In this section $\bot$ denotes a null stream record, rather than the all-negative selection outcome used in the unit-step algorithm.
The public vertex set is $V$ and the public horizon is $T\ge1$. A valid stream is $x_1,\ldots,x_T$ with each $x_t\in\binom V2\cup\{\bot\}$ and no edge inserted twice. Let $G_t$ contain all inserted edges through position $t$. Neighboring streams replace one true insertion at a fixed position by $\bot$, keeping the other positions unchanged. At every fixed prefix, their graphs differ by at most one edge.

At a public refresh time $s_*$, the static subroutine releases a candidate graph $C_{s_*}$. Until the next refresh, decoding is
\begin{equation}\label{eq:snapshot}
 M_t=E(G_{s_*})\cap C_{s_*},\qquad s_*\le t.
\end{equation}
The public output includes $s_*$. Endpoints use a saved adjacency snapshot or their edge arrival times to implement this intersection. This historical access determines the selected edges at the refresh time. Because the stream is insertion-only, $M_t\subseteq E(G_t)$, and the selected true edge set and its degree remain fixed between refreshes.

For any integer $k$, the statistic $\nu_k$ has edge sensitivity one. Removing a newly added edge from an optimal solution in the augmented graph leaves a feasible solution in the original graph, losing at most one edge. The reverse comparison gives the deletion direction.

\begin{theorem}[Arbitrary-order continual release]\label{thm:continual}
Fix $1\le k\le n-1$, $T\ge1$, $0<\theta\le1$, and a constant $s\in(0,1/4]$. Let
$K=1+\ceil{\log_{1+\theta/2}(n^2)}$ and $\lambda=1-2s$.
A central pure $\eps$-edge-DP mechanism guarantees, simultaneously for every $t\le T$ with probability at least $1-\rho$,
\[
 |M_t|\ge\frac{\nu_k(G_t)}{2+\theta}
       -O\left(\frac{K\log((T+K)/\rho)}{\eps}\right),
 \quad
 \Delta(M_t)\le\frac{k}{\lambda}
       +O\left(\frac{K\log(nK/\rho)}{\eps}\right).
\]
At time $t$, decoding uses the private graph snapshot at the most recent public refresh.
\end{theorem}

\subsection{Finite geometric epochs and exact capacity}
We use threshold monitoring for continual graph analysis \citep{fichtenberger2021,jain2024}. The geometrically spaced refresh schedule follows \citet[Section~4, Algorithm~4]{fichtenberger2021}; each trigger here releases a fresh candidate graph for snapshot decoding.
An \emph{epoch} is the interval monitored by one threshold mechanism, ending at its first positive answer. The thresholds grow geometrically, so a refresh occurs after the optimum has increased enough to make the old solution less informative. Here $\zeta$ is the fractional threshold increase, $K$ the deterministic epoch limit, $a$ the budget per monitor or static call, $L_s$ a tail parameter, $h_s$ the noise allowance, and $E_s$ the resulting threshold error allowance.
Write $\zeta=\theta/2$ and set
\begin{equation}\label{eq:stream-params}
\begin{split}
 K&=1+\ceil{\log_{1+\zeta}(n^2)},\qquad a=\eps/(2K),\\
 L_s&=\log\frac{8(T+2K)}\rho,\qquad h_s=4L_s/a,
 \qquad E_s=2h_s.
\end{split}
\end{equation}
The epoch thresholds are $\tau_j=(1+\zeta)^j$ for $j=0,\ldots,K-1$. Every epoch has a fresh single-hit $\AT_a$ monitoring $\nu_k(G_t)$. A positive response triggers a static invocation of privacy $a$ and failure probability $\rho/(2K)$, then advances to the next epoch, which may be queried at the same timestamp. At most $K$ epochs and $K$ static invocations occur on \emph{every} realization. If all epochs are exhausted, retain the final snapshot forever. Before any positive response, use the empty graph.

For the public fixed mask spacing $s\in(0,1/4]$, define
\begin{equation}\label{eq:stream-cap}
\begin{split}
 J_*&=\ceil{\log_{1+s}n},\qquad \lambda=1-2s,\\
 L_*&=\log\frac{64nK(n+J_*+3)}\rho,\\
 b&=\ceil{\frac{k}{\lambda}+\frac{256KL_*}{\eps}+2L_*}.
\end{split}
\end{equation}
This is a fixed public integer. A static invocation has internal local threshold budget $a/4$ and $h_*=16L_*/a=32KL_*/\eps$. Its required sequential slack is $8h_*+2L_*$, giving exactly \eqref{eq:stream-cap}. For fixed $s$, the cap is at most
$k/\lambda+O(K\log(nK/\rho)/\eps)$, including integer rounding. The bound is nontrivial when it is below the complete-graph degree $n-1$.

\begin{algorithm}[t]
\caption{Central continual coordination in arbitrary edge order}\label{alg:continual}
\KwIn{Public $V,T,\eps,\rho,k,\theta,s$, and \eqref{eq:stream-params}--\eqref{eq:stream-cap}; an insertion-only stream.}
Initialize $j\gets0$, $C_{\rm last}\gets\varnothing$, $s_*\gets0$, and $\AT_a$ at threshold $\tau_0=1$\;
\For{$t=1,\ldots,T$}{
Privately process the insertion or null record to form $G_t$\;
\While{$j<K$}{
Query the current $\AT_a$ on $\nu_k(G_t)$\;
\If{the answer is negative}{exit the while-loop\;}
Run a fresh geometric $\Seq(G_t)$ with privacy $a$, failure $\rho/(2K)$, cap $b$, and $\tau_{\rm sat}=b-6h_*$\;
Set $C_{\rm last}$ to its candidate graph and $s_*\gets t$; publish the candidate representation and refresh time\;
Set $j\gets j+1$\;
\If{$j<K$}{initialize a fresh $\AT_a$ at threshold $\tau_j$\;}
}
Output $(C_{\rm last},s_*)$, decoded against the row snapshot at $s_*$\;
}
If $s_*=0$, the decoder returns the empty edge set\;
\end{algorithm}

\subsection{Privacy of the output sequence}
View the whole fixed stream as the dataset. For a fixed public history, the prefix index queried by a monitor is fixed, and $\nu_k(G_t)$ has sensitivity one. Thus an entire epoch is $a$-DP by \Cref{lem:AT}. A static invocation at a transcript-determined refresh time receives neighboring prefix graphs, so its full public transcript is also $a$-DP.

A positive response ends its monitor before the static invocation begins, and the next monitor starts only after that invocation ends. The outer composition is therefore adaptive sequential composition of completed interactive mechanisms. At most $K$ monitors and $K$ static calls cost $2Ka=\eps$. The internal interleaving of each static call is already covered by the sequential privacy proof. Repeated outputs and the public refresh times are part of, or postprocessing of, this transcript. The deterministic epoch limit makes this budget bound hold on every realization.

\subsection{Simultaneous utility at all timestamps}
The same joint accuracy event will justify every timestamp, not just one preselected time. We first bound monitor errors and static-call failures, then compare the current optimum with the optimum at the most recent saved snapshot.
At most one negative monitor answer occurs per timestamp, and at most $K$ positive answers occur in total. There are therefore at most $T+K$ query noises and $K$ threshold noises. Each has scale at most $4/a$. Their union-bound failure probability is
\[
 (T+2K)e^{-L_s}=\rho/8\le\rho/2.
\]
Every fresh static call fails with conditional probability at most $\rho/(2K)$ on its prefix graph. A union bound over at most $K$ calls gives simultaneous static feasibility and utility except with probability $\rho/2$. Work on the intersection of these monitor and static accuracy events.

Let $q_t=\nu_k(G_t)$. A positive response at threshold $\tau_j$ implies $q_t\ge\tau_j-E_s$, while a negative response implies $q_t<\tau_j+E_s$. Before any positive answer, $q_t<1+E_s$ and the empty output obeys the desired utility bound with additive $E_s+1$.

Otherwise let $s_*$ be the most recent refresh. If epochs remain, let $j\ge1$ be the active epoch at the end of timestamp $t$. The last positive answer was at $\tau_{j-1}$ at time $s_*$, and timestamp $t$ ended with a negative answer at $\tau_j$. Therefore
\[
 q_t<\tau_j+E_s=(1+\zeta)\tau_{j-1}+E_s
 \le(1+\zeta)q_{s_*}+(2+\zeta)E_s.
\]
The snapshot is a static solution with $|M_t|\ge q_{s_*}/2$, so
\[
 |M_t|\ge\frac{q_t}{2(1+\zeta)}
       -\frac{2+\zeta}{2(1+\zeta)}E_s
 \ge\frac{q_t}{2+\theta}-E_s.
\]
If all $K$ epochs are exhausted, the last threshold satisfies $\tau_{K-1}\ge n^2\ge q_t$ for every valid graph. Its positive response at the final refresh implies $q_{s_*}\ge\tau_{K-1}-E_s\ge q_t-E_s$, which gives the same or a stronger lower bound. The same argument covers any exhaustion time on the bounded-noise event.

Finally, \eqref{eq:snapshot} keeps exactly the selected static edges until the next refresh. Consequently, simultaneously at all timestamps,
\[
 |M_t|\ge\frac{\nu_k(G_t)}{2+\theta}-(E_s+1),\qquad
 \Delta(M_t)\le b.
\]
Since $E_s=(16K/\eps)\log(8(T+2K)/\rho)$ and $b$ is \eqref{eq:stream-cap}, this proves \Cref{thm:continual}. The choice $\zeta=\theta/2$ is needed for the multiplicative factor $2+\theta$.

\subsection{Publicly delimited adjacency blocks}
This model supplies a complete neighbor list one vertex at a time, padded with absent-edge slots. A vertex's complete star means all input edges incident to it. Repeated endpoint records refer to one underlying edge, and their consistency is part of the input model.
The public order is $v_1,\ldots,v_n$. Each vertex $v$ has a fixed block of $n-1$ public slots indexed by the other vertices. Slot $(v,u)$ contains the private indicator of edge $\{v,u\}$, represented by that edge or a null record. The blocks have public start and end markers. A consistent stream gives the same indicator at $(v,u)$ and $(u,v)$. Neighboring streams change the two indicators of one underlying edge and no others. Thus adjacency is defined for one underlying relationship represented at its two endpoint positions, which may occur at different times.

Let $G_t$ contain every true edge seen in any slot by position $t$. After $j$ complete blocks, let $G^{(j)}$ be the union of the complete stars of $v_1,\ldots,v_j$. Buffer the current block privately. At its public end marker, run exactly geometric iteration $j$ of \Cref{alg:sequential}, with $\tau_{\rm sat}=b-6h$: its saturation tests and, if active, its proposal to later active vertices. Publish the updated transcript prefix at that end marker and repeat the previous output at other positions. The final end marker explicitly flushes the last block.

Every current candidate pair has a processed endpoint with a complete available block. The curator therefore knows its true-edge indicator and can evaluate all current saturation queries using the observed prefix. At a proposer's iteration, its entire neighborhood is available as well. Pair masks are addressed by public endpoint labels.

\begin{theorem}[Continual release by sequential prefixes]\label{thm:blocks}
Suppose vertices arrive in a public order with fixed public blocks of $n-1$ padded adjacency slots. The two records for any underlying edge are consistent; neighboring streams change both records of one edge. With the geometric sequential cap \eqref{eq:seq-cap}, a central pure $\eps$-edge-DP mechanism satisfies, simultaneously at every stream position with probability at least $1-\rho$,
\[
 \Delta(M_t)\le b,\qquad
 |M_t|\ge\nu_k(G_t)/2-k/2.
\]
At every complete-block boundary, including the final flushed boundary, the additive $k/2$ loss disappears.
\end{theorem}

\begin{proof}[Proof of \Cref{thm:blocks}]
Fix the final graph $G$ and couple the streaming execution to \Cref{alg:sequential} on $G$ using the same order, pair masks, and private noises. At each block boundary, every quantity used by the sequential iteration is available to the stream algorithm and agrees with the static execution. The boundary transcripts are therefore prefixes of one sequential transcript. Repetitions occur at fixed public positions. Neighboring consistent streams correspond to edge-neighboring final graphs, so privacy of the complete sequential transcript proves privacy of the entire continual sequence. Repeated releases at public positions are postprocessing of that single transcript.

On the single sequential accuracy event $\mathcal E$, all prefixes are feasible, and after $j$ blocks every unselected edge of $G^{(j)}$ has an endpoint of current selected degree at least $k$. This is the prefix invariant of \Cref{thm:sequential}. All selected true edges $M^{(j)}$ lie in $G^{(j)}$, so \Cref{lem:charging} gives
\[
 |M^{(j)}|\ge\nu_k(G^{(j)})/2.
\]
During the next partial block, every newly seen edge not already in $G^{(j)}$ is incident to $v_{j+1}$. A $k$-matching uses at most $k$ such edges. Removing them gives a feasible $k$-matching in $G^{(j)}$, proving
\[
 \nu_k(G_t)\le\nu_k(G^{(j)})+k.
\]
The output is still $M^{(j)}$, hence $|M_t|\ge\nu_k(G_t)/2-k/2$. At a completed-block boundary the partial-block loss is absent, including at the final flush. The single event $\mathcal E$ establishes these statements simultaneously over all stream positions.
\end{proof}

At release, each candidate star's proposer has supplied its complete block. Later consistent records preserve the truth values of all pairs in that star, so its decoded intersection stays fixed. This property permits current-row decoding at the public padded-block boundaries, in contrast to refresh-time snapshots in arbitrary edge order.
\section{Examples Separating Structural Guarantees}\label{app:counterexamples}
\paragraph{Edge count and matching containment.}
For $q\ge3$, let $G$ have edges $\{c,\ell_i\}$ and $\{\ell_i,x_i\}$ for $1\le i\le q$, with all named vertices distinct. Select only $M=\{\{c,\ell_i\}:1\le i\le q\}$. Then $|M|=q=\nu_1(G)$, but $\nu_1(M)=1$. Every omitted edge touches a vertex $\ell_i$ of selected degree one, so the benchmark-one charging invariant also holds. By contrast, \Cref{thm:main} uses saturation at a level $k$ close to $b$, which yields matching containment through \Cref{lem:kernel}. In this example, any single star edge can be extended by another disjoint edge $\{\ell_i,x_i\}$ in $G$, so the selected graph contains no maximal matching of the input.

\paragraph{The limiting factor of saturation.}
Take $H=K_{b,b}$, the complete bipartite graph with $b$ vertices on each side and every cross-side edge, and construct $G$ by attaching one new leaf to each of the $2b$ vertices of $H$. Every omitted edge touches a vertex of degree $b$ in $H$, yet $\nu_1(H)=b$ and $\nu_1(G)=2b$. Thus degree at most $b$ and saturation at $b$, considered as a combinatorial certificate, permit a matching-preservation factor of exactly two. This example identifies the limit of that certificate independently of the privacy mechanism.

\paragraph{Scalar fidelity and edge overlap.}
Privately estimate $\nu_1(G)$ with sensitivity-one Laplace noise, round and clip the estimate to $\{0,\ldots,\floor{n/2}\}$, and publish that many disjoint pairs in a fixed public label order. The synthetic graph has exactly the estimated matching size, while its pairs may differ from the input edges. This separates scalar fidelity from the overlap objective in \Cref{thm:explicit}.

\paragraph{Current intersections and snapshot feasibility.}
A public candidate graph may contain many non-edges incident to one vertex. Inserting those pairs increases the current intersection degree even though the candidate graph is unchanged. In contrast, a fixed selected true edge set keeps its degrees under insertions. The refresh-time decoder in \eqref{eq:snapshot} uses this latter property to preserve feasibility between releases.
\section{Resource Accounting and the Shortlist Workflow}\label{app:resources}
We derive the resource bounds in \Cref{tab:resources} for the specified implementations. All public parameters, vertex identifiers, orders, and stream positions are fixed and known. The counts apply to every realization of the protocol.

\subsection{Resource definitions}
A public bit is one bit written to the common billboard. The communication measure includes public mask labels and protocol responses. It counts a single shared copy; network fanout, packet headers, authentication, private input transfer to a curator, and transmission delays are separate costs. A threshold query is one evaluation supplied to AboveThreshold. Evaluating a query may involve many adjacency probes or arithmetic operations, especially when the query computes a matching value.

Stored size is the length of an encoding of the complete logical public output history. Private noises and unreleased curator state stay outside that history. Reconstructing the history from its encoding has its own computation and workspace costs. For streams, a public clock identifies positions. The table uses a transport that emits a new description when the output changes and a repeat flag otherwise, adding at most $O(T)$ online bits. This convention distinguishes producing the public sequence from storing the information needed to replay it.

For the geometric scheme, let
\[
 J_g=\ceil{\log_{1+s}n},\quad
 \ell_g=\ceil{\log_2(J_g+3)},\quad
 F_g=n^2\ell_g,\quad Q_g=n^2+n(J_g+2).
\]
The label alphabet has size at most $J_g+1$; using $J_g+3$ also accommodates stopping and default symbols. For the parallel scheme, let
\[
 J_p=\ceil{\log_2 n},\qquad
 F_p=Rn^2\ceil{\log_2(J_p+3)}.
\]
Pairs are addressed in a fixed public order, so their positions identify their endpoints. The spacing in $J_g$ is that of the corresponding construction: $s=\eta/16$ for same-capacity geometric selection and fixed $s$ for the continual mechanisms.

\subsection{Sequential and parallel counts}
\paragraph{Unit-step.}
At most $n$ active-vertex saturation tests occur in each of $n$ iterations, producing at most $n^2$ public answer bits. Each vertex proposes at most once. An invoked proposer returns an index in $\{0,\ldots,n-1\}$ or the all-negative default. Including the skipped symbol, a fixed-width representation uses $\ceil{\log_2(n+2)}$ bits per scheduled vertex. Thus
\[
 B_{\rm online}\le n^2+n\ceil{\log_2(n+2)}=O(n^2).
\]
The protocol evaluates at most $n^2$ PVSM queries and $n^2$ saturation queries, with at most $n$ of each type per vertex. Public vertex order determines the masks. After the interaction, \Cref{prop:compact} encodes the full history in the $O(n\log n)$ bits of \eqref{eq:compact-bits}, using only public data. A constant number of synchronizations per iteration gives $O(n)$ rounds. The distinction between online communication and stored size comes from replacing repeated saturation answers by one final deactivation time per vertex.

\paragraph{Geometric.}
Across all proposals there are at most $\binom n2$ potential-pair labels, because only the earlier endpoint can propose to the later one. Including stopping messages and persistent saturation answers gives
\[
 B_{\rm online}\le n^2+\bigl[\tbinom n2+n\bigr]\ell_g=O(F_g).
\]
The direct message sequence or its fixed-address representation uses $O(F_g)$ stored bits. Each vertex performs at most $J_g+2$ PVSM queries and $n$ saturation queries, giving $Q_g$ in total. The protocol uses $O(n)$ synchronized rounds. Geometric masks therefore shorten proposal selection, while the persistent saturation tests retain their $n$-iteration horizon.

\paragraph{Parallel.}
Each of $R$ phases uses at most $n$ role bits, two independent arrays of at most $\binom n2$ pair labels, at most $n$ PVSM outcomes, and at most $n$ saturation bits. Each active vertex has one role per phase. Writing $\ell_p=\ceil{\log_2(J_p+3)}$, we obtain
\[
 B_{\rm online}\le R\{2n+[2\tbinom n2+n]\ell_p\}=O(F_p).
\]
The full message history has the same storage bound. At most $nR(J_p+2)$ PVSM queries and $nR$ saturation queries give $nR(J_p+3)$ evaluations, or at most $R(J_p+3)$ per vertex. A constant number of synchronizations per phase gives $O(R)$ rounds. The bound permits both complete label arrays; an implementation may omit unused labels. All three static constructions decode from a vertex's static adjacency row and the public history.

\subsection{Continual-release counts}
\paragraph{Arbitrary edge order.}
The curator evaluates at most $T+K$ monitor queries and executes at most $K$ static geometric invocations, for a total of at most
\[
 T+K+KQ_g.
\]
threshold-query evaluations. A monitor query evaluates $\nu_k(G_t)$ centrally. Its cost includes solving the specified matching-value problem. The at-most-$K$ refreshes are release events, possibly with several at the same timestamp; the resource measure here is a central continual computation.

Retaining the static logical transcript for each refresh costs $O(F_g)$ bits per refresh. This bounds storage of both the full static interaction and its candidate representation. Ordered refresh times and a never-used symbol use at most $K\ceil{\log_2(T+2)}$ bits. Therefore, with
\[
 F_c=K\bigl(F_g+\ceil{\log_2(T+2)}\bigr),
\]
the complete stored public history uses $O(F_c)$ bits. Refresh times identify the current output at every timestamp and reconstruct the monitor's logical responses. Each refresh corresponds to a positive answer ending an epoch; each processed timestamp before exhaustion ends with a negative answer. After $K$ epochs the monitor stops. Multiple refreshes at one timestamp retain their epoch order. This reconstruction uses only the public record and applies on every realization.

With repeat flags, the online cost is $O(T+F_c)$ bits, including fresh descriptions and timestamps. A transport with silence for unchanged outputs can omit the $T$ flags. Literal retransmission of the current description can instead cost $O(T(F_g+\log(T+2)))$ plus refresh descriptions. Each transport represents the same sequence of outputs. Endpoints decode using a saved refresh-time row or edge arrival times.

\paragraph{Public adjacency blocks.}
At each of the $n$ public block boundaries, the curator executes one geometric sequential iteration. Thus at most $Q_g$ threshold queries are evaluated, and the complete boundary transcript has the static protocol's $O(F_g)$ encoding. Public positions determine intermediate repetitions. The repeat-flag transport uses $O(T+F_g)$ online bits, where $T$ includes the slots and boundary markers; the fixed padded blocks have $T=\Theta(n^2)$.

Each released star has a proposer whose block is complete. All of its candidate pairs' truth values are therefore known when the star is published, and later consistent records preserve those values. Current observed rows suffice for decoding in this public padded-block model.

\paragraph{Full parameter-dependent resource table.}
The quantities $F_g$ and $F_p$ above are direct public-description sizes for geometric and parallel masks; $Q_g$ counts geometric threshold evaluations. The quantity $F_c$ adds the descriptions and timestamps over all refreshes. \Cref{tab:resources-exact} collects these exact-parameter bounds. Substituting fixed accuracy, polynomially small failure probability, and a polynomial horizon gives the simplified rates in \Cref{tab:resources}.
\begin{table}[t]
\centering
\small
\setlength{\tabcolsep}{3pt}
\renewcommand{\arraystretch}{1.13}
\begin{tabularx}{\textwidth}{@{}>{\raggedright\arraybackslash}p{0.17\linewidth}>{\raggedright\arraybackslash}p{0.14\linewidth}>{\raggedright\arraybackslash}p{0.25\linewidth}>{\raggedright\arraybackslash}p{0.21\linewidth}>{\raggedright\arraybackslash}X@{}}
\toprule
Protocol;\newline privacy & Rounds or refreshes & Public bits:\newline online / stored & Threshold queries (total) & Private decoder access \\
\midrule
Unit-step;\newline local & $O(n)$ rounds & $O(n^2)$ / $O(n\log n)$ & $\le 2n^2$ & Static row \\
Geometric;\newline local & $O(n)$ rounds & $O(F_g)$ / $O(F_g)$ & $\le Q_g$ & Static row \\
Parallel;\newline local & $O(R)$ rounds & $O(F_p)$ / $O(F_p)$ & $\le nR(J_p+3)$ & Static row \\
Arbitrary order;\newline central & $\le K$ refreshes & $O(T+F_c)$ / $O(F_c)$ & $\le T+K+KQ_g$ & Refresh-time row \\
Adjacency blocks;\newline central & $n$ block\newline boundaries & $O(T+F_g)$ / $O(F_g)$ & $\le Q_g$ & Observed row\newline(completed stars) \\
\bottomrule
\end{tabularx}
\caption{Resource bounds for the direct constructions. Online bits count public communication; stored bits encode the complete logical output history. Threshold queries count evaluations rather than arithmetic operations. Here $J_g=\ceil{\log_{1+s}n}$, $J_p=\ceil{\log_2 n}$, $R=O(\log(n/\rho))$, and $K$ is defined in \Cref{thm:continual}. Put $F_g=n^2\ceil{\log_2(J_g+3)}$, $F_p=Rn^2\ceil{\log_2(J_p+3)}$, $Q_g=n^2+n(J_g+2)$, and $F_c=K(F_g+\ceil{\log_2(T+2)})$. Public mask generation is included in communication. For the geometric same-capacity result $s=\eta/16$; the stated continual constructions use fixed $s$.}
\label{tab:resources-exact}
\end{table}

\subsection{Reciprocal standby directories}\label{app:workflow}
The illustrative workflow uses a fixed undirected eligibility graph among computing sites. Each edge is a pre-existing bilateral arrangement known to both endpoints. The planning task is to retain bounded-size, reciprocal directories of eligible standby partners. Each directory is a local record of assignment options; selecting or activating a particular assignment is a later stage.

\paragraph{Inputs and coordination.}
Each site holds its eligibility row. Identifiers, vertex order, capacity, and privacy and failure parameters are public. The local protocols access each row at its owning site. The coordinator and other observers receive the prescribed flags, stopping indices, and any public mask labels, forming $B$. Degree counts, threshold noises, and decoded incident edges remain at the site. Each site obtains its directory by intersecting its row with the public candidates.

\paragraph{The completed planning task.}
A site marks its decoded partner identifiers in its existing records. This local update completes the task: reciprocal selection follows from the shared candidate graph, so completion requires no selected-pair acknowledgement or directory upload. On the success event, every directory has at most $b$ entries. The edge-count guarantee retains a large collection of eligible relationships, while matching containment certifies that a large collection of disjoint bilateral pairs remains available within those directories. These guarantees concern the fixed eligibility graph used during coordination.

\paragraph{Protected observations.}
For a changed edge $\{u,v\}$, the public transcript and the joint planning outputs of all other sites satisfy the same DP comparison by \Cref{lem:billboard}. Their unchanged eligibility rows can be included as well. The two endpoints already hold their relationship as input. This specifies the observations protected at the end of planning.

\paragraph{Later activation.}
A subsequent request to a selected partner communicates the request and selected status to its recipient. Depending on the communication system, other observers may learn recipient identity, timing, or contents. Those messages use private selected edges in addition to the billboard, so an activation protocol requires its own observation model and privacy analysis. Directory production and assignment activation are distinct stages of the workflow.

\subsection{Accuracy dependence of the saturation certificate}\label{app:eta-discussion}
With $D=D_4$ fixed independently of $\eta$, the coefficient used for both structural utilities is
\[
 \frac{b-D}{2(b+1)}=\frac12-\frac{D+1}{2(b+1)}.
\]
The target $1/(2+\eta)=1/2-\eta/[2(2+\eta)]$ is certified exactly when
\[
 (2+\eta)(D+1)\le\eta(b+1),
 \qquad\text{equivalently}\qquad
 \eta b\ge(2+\eta)D+2.
\]
The linear $1/\eta$ dependence therefore converts an additive saturation deficit into the relative accuracy required by the certificate. Below this sufficient capacity, the same coefficient leaves the desired utility unresolved. The star-based degree lower bound uses single-edge success for every finite multiplicative factor; its dependence is on privacy and failure probability. The complete-bipartite-plus-leaves example in \Cref{app:counterexamples} identifies the factor-two limit of saturation as a combinatorial condition. Determining the output model's optimal accuracy dependence calls for either a stronger utility analysis or an accuracy-sensitive lower bound.
\section{Notation and Reading Guide}\label{app:notation}
This appendix collects symbols introduced in the main text and proofs. Each proof also defines its local variables. The same letter can denote different objects in different subsections; the entries below identify those uses explicitly.

\subsection{Graphs, outputs, and approximation}
\begin{longtable}{@{}p{0.25\textwidth}p{0.71\textwidth}@{}}
\toprule
Symbol & Meaning \\
\midrule
\endhead
$G=(V,E)$, $n$, $m$ & Input graph with public vertex set $V$ and private edge set $E$; $n=|V|$ and $m=|E(G)|$. \\
$\binom V2$ & Set of all unordered pairs of distinct vertices, not just input edges. \\
$N_G(v)$ & Neighbors of vertex $v$; its private adjacency list or row. \\
$B$ & Complete public transcript, including messages and public random choices. \\
$C(B)$, $C(v)$ & Candidate graph determined by the transcript; candidate neighbors of $v$. During a run, $C$ is the current candidate graph. \\
$M$, $M(v)$ & Selected true edges $E(G)\cap C(B)$; the selected-neighbor list at $v$. \\
$|F|$, $d_F(v)$, $\Delta(F)$ & Number of edges, degree at $v$, and maximum degree of an edge set $F$. \\
$\nu_k(G)$ & Maximum number of edges in a subgraph whose degree is at most $k$ at every vertex. $\nu_1(G)$ is the maximum ordinary matching size. \\
$b$, $k$ & Output capacity; comparison capacity or degree target for saturation. They need not be equal. \\
$H\subseteq G$ & A subgraph on the same vertex set; $G\setminus H$ means the omitted edges. In the parallel proof, $H$ is the residual graph. \\
\bottomrule
\end{longtable}
An ordinary matching uses each vertex at most once. A $k$-matching uses it at most $k$ times. A maximal matching cannot be enlarged by adding one more edge; a maximum matching has the largest possible size. A matching kernel here is a bounded-degree subgraph retaining a provably large ordinary matching. Saturation at level $k$ means every omitted edge has an endpoint with selected degree at least $k$.

\subsection{Privacy, probability, and asymptotic notation}
\begin{longtable}{@{}p{0.25\textwidth}p{0.71\textwidth}@{}}
\toprule
Symbol & Meaning \\
\midrule
\endhead
$\eps$ & Overall edge-privacy budget; smaller values give a stronger distributional constraint. \\
$\delta$ & Additive allowance in approximate DP; all positive algorithms here use pure privacy, $\delta=0$. \\
$\rho$ & Allowed probability that a degree or utility guarantee fails. It is separate from $\delta$. \\
$\eta$, $\theta$ & Accuracy parameters giving factors $2+\eta$ in the static main result and $2+\theta$ in arbitrary-order continual release. \\
$a$, $a_v$ & Budget of one mechanism invocation; total local budget at vertex $v$. \\
$\Prb$, $\E$ & Probability and expectation over the specified random choices. \\
$\ind_A$, $\ind\{A\}$ & Indicator: one when event $A$ holds, zero otherwise. \\
$\mathcal E$, $\mathcal E^c$ & Joint accuracy event and its complement. Subscripts specify a variant or a prefix. \\
$\Lap(t)$ & Mean-zero Laplace distribution with positive scale $t$. \\
$\mathrm{Ber}(p)$, $\mathrm{Bin}(\ell,p)$ & A zero--one trial with success probability $p$; the sum of $\ell$ independent such trials, for a nonnegative integer $\ell$. \\
$\mathcal F_t$ & Information exposed for the proof through phase $t$ (a filtration). \\
\bottomrule
\end{longtable}
The notation $O(f)$ denotes an asymptotic upper bound up to a constant multiple, $\Omega(f)$ a lower bound, and $\Theta(f)$ both. Unless a base is specified, $\log$ is the natural logarithm; $\log_2$ is used for bit counts. The ceiling $\ceil{x}$ is the least integer at least $x$, and the floor $\floor{x}$ the greatest integer at most $x$. All probabilities refer to a fixed input graph or stream unless an explicit comparison is made.

\subsection{Selection and saturation}
\begin{longtable}{@{}p{0.25\textwidth}p{0.71\textwidth}@{}}
\toprule
Symbol & Meaning \\
\midrule
\endhead
$\AT_a$ & Single-hit AboveThreshold with budget $a$, stopping at its first positive response. \\
PVSM & Public Vertex Subset Mechanism: privately chooses a public candidate mask. \\
$S$, $A_r$ & Public eligible new neighbors; nested public mask at index $r$. \\
$x$, $Z_r$ & Selected degree before a PVSM invocation; true-neighbor count in $A_r$. The degree after selecting $A_r$ is $x+Z_r$. \\
$q_r$ & Query supplied to the threshold mechanism. PVSM uses $q_r=-(x+Z_r)$. \\
$h$, $h_u$ & Bounds on noise magnitude used to define the accuracy event; $h_u$ is for unit-step masks. These are allowances, not sampled noises. \\
$L$, $L_u$ & Logarithmic parameters used in the corresponding failure-probability bounds. \\
$\tau_{\rm sat}$ & Public base threshold for the persistent saturation test. \\
$D_4$, $D_8$ & Fixed integer degree allowances $\ceil{4h_u}$ and $\ceil{8h_u}$ for two saturation thresholds. \\
$s$, $\lambda$ & Geometric-mask spacing and retention coefficient $\lambda=1-2s$. \\
$J$ & Last random-mask index. It is $\ceil{\log_{1+s}n}$ for geometric masks and $\ceil{\log_2n}$ for dyadic masks. \\
$U_{vu}$ & Public uniform random value defining nested masks for a potential pair. \\
$E_0$ & Geometric construction's additive capacity allowance $8h+2L$. \\
$\bot$ & Distinct all-negative outcome for unit-step selection; in the streaming section, a null input record. \\
\bottomrule
\end{longtable}
In the threshold lemma $D$ denotes the private dataset and $T$ a base threshold. In the finite-encoding lemma $D$ instead denotes a random integer mask label. In the matching-kernel proof $D$ is an auxiliary multigraph. The local symbol $D_r$ in the thinning proof counts deleted neighbors; it is unrelated to the fixed allowance $D_4$. The proof of the simple capacity constant uses $H=\log(n/\rho)$ locally, rather than a graph.

\subsection{Lower bounds, parallel phases, and streams}
In the explicit lower bound, $X_e\in[0,1]$ is the public weight on potential edge $e$, $U_G(X)$ its total weight on true edges, and $\bar b$ the average expected-degree budget. The factor $\gamma$ and additive loss $\beta$ specify the assumed utility. The allowance $\delta_4=\delta\sum_{j=0}^3e^{j\eps}$ accounts for four edge edits. In the padding argument, $n_{\mathrm{even}}=2\floor{n/2}$ is the number of non-isolated vertices. In the star lower bound, $q$ counts leaves. In the decoder representation proof, $D_v$ is a fixed deterministic local decoding rule, $\mu_G$ the transcript law on input $G$, and $F_*$ a common set of transcripts of probability one.

For the parallel construction, $R$ is the number of phases and $V_*$ the active vertices. The sets $P$ and $Q$ contain proposers and receivers. The sets $P_v$ and $Q_w$ are their public proposal and acceptance sets. The constants $a_0$ and $a_1$ are persistent and per-phase privacy budgets. The analysis uses reference-count target $K_0=384h$ and reference inclusion probabilities $p_v$. The stopped potential $Z_t$ counts residual edges only while the required accuracy conditions have held. The joint success event $\mathcal S$ additionally requires an empty final residual graph. This $Z_t$ is different from a mask count $Z_r$.

For streams, $T$ is the public number of positions, $t$ a position, $G_t$ the graph seen so far, and $M_t$ the decoded selected edge set. The most recent refresh is at $s_*$; decoding uses the snapshot $G_{s_*}$. The deterministic refresh limit is $K$. The epoch index is $j$, the growth factor is $1+\zeta$ with $\zeta=\theta/2$, and $\tau_j=(1+\zeta)^j$ is the next threshold. Monitor parameters $L_s,h_s,E_s$ are its tail, noise, and error allowances. Static-refresh parameters $J_*,L_*,h_*$ belong to the geometric subroutine. The current monitored optimum is $q_t=\nu_k(G_t)$. For adjacency blocks, $G^{(j)}$ is the union of the complete stars observed after $j$ blocks.

The resource appendix defines $J_g,J_p$ as geometric and parallel mask horizons, $\ell_g,\ell_p$ as label bit lengths, $Q_g$ as the geometric query count, and $F_g,F_p,F_c$ as public-description size bounds. These quantities count messages and evaluations, not wall-clock time, private workspace, or network fanout.
\label{document:lastpage}
\end{document}